\documentclass[final,onefignum,onetabnum]{siamart190516}

\usepackage{graphicx}
\usepackage{dcolumn}
\usepackage{bm}
\usepackage{color,xcolor}
\usepackage{bm}
\usepackage{url}
\usepackage{mathtools}
\usepackage{float}
\usepackage{braket}
\usepackage{amsmath}	
\usepackage{amssymb}	
\usepackage{algorithmic}	

\usepackage{cite}

\newtheorem{thm}{Theorem}[section] 
\newtheorem{remark}[thm]{Remark}

\DeclareMathOperator{\Tr}{Tr}

\title{Ranking edges  by their impact on the spectral complexity of information diffusion over networks\thanks{JK and DT were supported in part by the National Science Foundation (DMS-2052720 and EDT-1551069) and the Simons Foundation (grant \#578333). PJM was supported in part by the National Science Foundation (BCS-2140024).
}}

\author{
Jeremy Kazimer\thanks{Department of Mathematics, University at Buffalo, State University of New York (SUNY), Buffalo, NY 14260, USA}
\and
Manlio De Domenico\thanks{CoMuNe Lab, Department of Physics and Astronomy "Galileo Galilei", University of Padua, 35131, Italy}
\and
Peter J. Mucha\thanks{Department of Mathematics, Dartmouth College, Hanover, NH, 03755, USA}
\and
Dane Taylor\thanks{{School of Computing \& Department of Mathematics and Statistics, University of Wyoming, Laramie, WY 82072, USA and} Department of Mathematics, University at Buffalo, State University of New York (SUNY), Buffalo, NY 14260, USA {(dane.taylor@uwyo.edu)}}
}

\begin{document}
\maketitle


\begin{abstract}
Despite the numerous ways now available to quantify which parts or subsystems of a network are most important, there remains a lack of  centrality measures that  are related to the complexity of information flows and   are derived directly from entropy measures. Here, we introduce a ranking of  edges based on how each edge's removal would change a system's von Neumann entropy (VNE),  which is a spectral-entropy measure that has been adapted from quantum information theory  to quantify the complexity of information  dynamics over networks. 
{We show that a direct calculation of such rankings is computationally inefficient (or unfeasible) for large networks,  since the possible removal of $M$ edges requires that one  computes all the eigenvalues of  $M$ distinct matrices.} 
To overcome this limitation, we employ spectral perturbation theory to estimate VNE perturbations and derive an approximate edge-ranking algorithm that is {accurate and has a computational complexity that scales   as $\mathcal{O}(MN)$ for networks with $N$ nodes.}
Focusing on a form of VNE that is associated with a transport operator $e^{-\beta{  L}}$, where  ${  L}$ is a graph Laplacian matrix and $\beta>0$ is a diffusion timescale parameter, we apply this approach to diverse applications including a network encoding polarized voting patterns of the 117th U.S. Senate, a multimodal transportation system including roads and metro lines in London, and a multiplex brain network encoding correlated human brain activity. Our experiments highlight situations where the edges that are considered to be most important for  information diffusion complexity can dramatically change as one considers short, intermediate and long  timescales $\beta$ for diffusion.
%
\end{abstract}


\begin{keywords}
network science,
perturbation theory,
information diffusion,
spectral complexity, 
von Neumann entropy,
edge ranking
\end{keywords}

\begin{AMS}
94A17, 05C82, 60J60, 28D20, 68P30
\end{AMS}

\section{Introduction}\label{sec:Intro}

Centrality analysis \cite{borgatti2005centrality,newman2010networks} --- i.e., determining the importance of substructures including nodes
\cite{brandes2001faster,estrada2009communicability,gleich2015pagerank}, edges \cite{girvan2002community,restrepo2006,tong2012gelling,zhang2017edgecentrality}, and subgraphs \cite{estrada2005subgraphs,taylor2011network} --- is a fundamental pursuit of network science. Quantifying their importance supports diverse applications ranging from the ranking of webpages \cite{brin1998anatomy,kleinberg2011authoritative}, athletes/athletic teams \cite{saavedra2010mutually,callaghan2003random,park2005network}, and academics/academic institutions \cite{schmidt2007ranking,ding2009pagerank,clauset2015systematic,taylor2017eigenvector}, to the identification of potential intervention targets for dynamical processes over networks: examples include
congestion points for transportation systems \cite{holme2003congestion,guimera2005worldwide}, 
influencers in social networks \cite{kempe2003maximizing,kitsak2010identification,lehmann2018complex},
vaccination strategies  \cite{shaw2010enhanced,rushmore2014network}, as well as
points of fragility for  critical infrastructures \cite{hines2008centrality,restrepo2006} and biological networks \cite{jeong2001lethality,allesina2009googling}. Given the broad impact of this interdisciplinary field, new methodological advances continue to be developed including the extension of centralities to generalized types of networks (e.g., temporal networks \cite{pan2011path,taylor2017eigenvector}, multilayer and multiplex networks \cite{de2015structural,de2015ranking,sole2016random,tudisco2018node,deford2018new,taylor2019tunable}, hypergraphs and simplicial complexes \cite{estrada2018centralities,benson2019three,tudisco2021node}), comparing various centralities for specific applications \cite{jordan2007quantifying,takes2016centrality}, deriving centralities that cater  to particular dynamical systems \cite{restrepo2006,taylor2016synchronization}, and building stronger theoretical foundations using tools from  statistics \cite{de2018physical}, mathematics \cite{estrada2010network}, and machine learning \cite{grando2018machine}.

Herein, we aim to further align the study of centrality with information theory and entropy measures, and in particular, to contribute to the growing  interface between network science and quantum information theory \cite{passerini2008neumann,anand2011shannon,de2015structural,de2016spectral,li2018network,ghavasieh2020enhancing,ghavasieh2020statistical}.  For a recent overview of the topic, we refer the reader to \cite{ghavasieh2022statistical}.  Our approach is partly motivated by node entanglement    \cite{ghavasieh2021unraveling}, which was recently proposed as a node centrality measure stemming
from the study of a network's von Neumann entropy (VNE) \cite{neumann2013mathematische}. 
VNE was originally developed as an information-theoretic measure to quantify disorder within quantum-mechanical systems, and it has played a crucial role in the development of quantum information theory \cite{neumann2013mathematische,wilde2013quantum}. More recently, it has been extended as a spectral-entropy measure that quantifies the structural  complexity of graphs  \cite{braunstein2006laplacian,de2015structural,li2018network}, has been generalized to preserve the sub-additivity property  \cite{de2016spectral} {as later clarified in Sec.~\eqref{sec:VNE}}, and has been applied to a variety of interconnected systems, from molecular biology \cite{ghavasieh2021multiscale} to neuroscience \cite{benigni2021persistence}. 
Of particular relevance is previous research showing that VNE is a natural measure to   quantify the   complexity of the interplay between structure and information diffusion over a network \cite{ghavasieh2020statistical,ghavasieh2022statistical}.
We also note in passing that there exist other  centrality measures \cite{nikolaev2015centrality,chellappan2013wireless,qiao2018entropy,li2020entropy,tutzauer2009entropybasedca, tutzauer2007pathtransfer} that also utilize entropy in some way, but which are not derived from VNE  nor   relate to the complexity of  information diffusion.
{Generally speaking, the literature relating edge centrality measures to entropy is lacking, and in particular,   no centrality measure for edges  has been previously derived using von Neumann entropy.}



In this work, we develop
a  VNE-based measure for edge importance.
(Although our methods  easily extend to nodes and subgraphs, we will not explicitly use them so here.)
%
%
We develop a framework that ranks edges based on how VNE would change if each edge were separately removed. %
We interpret these rankings from the perspective of information diffusion over networks \cite{ghavasieh2020statistical,ghavasieh2022statistical} so that the node rankings are directly associated with the spectral complexity of information diffusion. The top-ranked edge is the one whose removal would most increase VNE, thereby most increasing the spectral complexity of information diffusion.
As such, our approach complements the numerous existing centralities that relate  to information-spreading dynamics over   networks  including betweenness \cite{freeman1977set}, Katz centrality \cite{katz1953new}, communicability \cite{estrada2012physics,estrada2009communicability}, and many others \cite{estrada2010network,kempe2003maximizing,kitsak2010identification,lehmann2018complex,bertagnolli2020network}. These existing methods typically aim to identify nodes/edges that are important as either sources for broadcasting information or are crucial bottlenecks for information pathways, and as such, they do not provide a notion of centrality as it relates to the complexity/heterogeneity of \emph{information diffusion}.
Our work highlights how the extension of some concepts used in quantum information theory to network science, and centrality analysis in particular, can provide complementary insights toward the importance of network substructures (e.g., edges) with respect to a system's functional complexity and how their roles change as one considers different   timescales for information diffusion.
%

Our methodology is also motivated by prior centrality analyses \cite{restrepo2006,milanese2010approximating,taylor2011network,tong2012gelling} that utilize spectral perturbation theory to more efficiently estimate the impact of node, edge, or subgraph removals on the perturbed spectral properties of networks (or functions thereof \cite{taylor2016synchronization}).
Specifically, as a spectral-entropy measure, the computation of  VNE requires calculation of the full set of eigenvalues for a graph's Laplacian matrix. (See Sec.~\ref{sec:Back} for our general definition.) {Computing all the eigenvalues of a large matrix is well-known to be computationally expensive, and so} naively recomputing the eigenvalues to obtain the new VNE that occurs after an edge removal is computationally {demanding} and is infeasible for large networks. {(See Sec.~\ref{sec:runtime} for numerical evidence and a discussion.)}
Thus motivated, we employ spectral perturbation theory to efficiently approximate   changes to VNE and approximate the subsequent ranking of edges based on those  perturbations. We show that the computational complexity of the approximate rankings 
%
{in the limit of large $N$} reduces to {at most} $\mathcal{O}(N)$ per edge, allowing the framework to effectively scale to large networks.  

We investigate the true and approximate  rankings of edges in terms of VNE change for three empirical networks: 
a voting similarity network for the 117th U.S. Senate, 
a multimodal transportation network encoding roads and metro lines, and 
a multiplex brain network with two layers that represent correlated brain activity at different frequency bands.
We focus on a version of VNE based on a transport operator $e^{-\beta{  L}}$, where  ${  L}$ is a graph Laplacian matrix and $\beta>0$ is a  \emph{diffusion timescale parameter}.
Interestingly, we find that our measure for a network's most important edges is significantly influenced by $\beta$
in that the top-ranked edges at a short timescale are often very different from those at an intermediate and/or long timescale. This result is interesting because it highlights the multiscale nature of importance for fundamental units of a network. 
Our numerical experiments reveal three different types of network structures that systematically cause such a reordering effect:
community structure in the U.S. Senate voting network;
the relative speed between roads and metro lines for the transportation network; 
 and the strength of coupling between network layers for the multiplex brain network.
These experiments also highlight the potential broad utility of our proposed entropy-based rankings for diverse applications across the social, physical, and biological sciences. 

The remainder of this paper is organized as follows:
we present background information in Sec.~\ref{sec:Back},
methodology in Sec.~\ref{sec:results},
numerical experiments in Sec.~\ref{sec:emp_exp}, and
a discussion in Sec.~\ref{sec:Discuss}.

\section{Background information}\label{sec:Back}

We first provide  background information Laplacian matrices and their spectral properties (Sec.~\ref{sec:notation}), VNE as a measure of spectral complexity for information diffusion  (Sec.~\ref{sec:VNE}), and prior centrality measures that stem from spectral   perturbation theory (Sec.~\ref{sec:pert0}).

\subsection{Laplacian matrices and  the diffusion equation}\label{sec:notation}
%
Let {${G}(\mathcal{V},\mathcal{E})$} be a graph with a set $\mathcal{V}=\{1,\dots,N\}$  of nodes and a set {{$\mathcal{E} \subset\mathcal{V}\times\mathcal{V}$}} of undirected edges, each of which can have a positive weight $W_{ij}>0$.
We assume that there are no self-edges. The graph can be equivalently defined by a symmetric adjacency matrix {$A$ with entries $A_{ij}=W_{ij}$ if $(i,j)\in\mathcal{E}$ and $A_{ij}=0$ otherwise.} 
%
We define  $M=\frac{1}{2}\sum_{ij} A_{ij}$ as the total weight of edges. (The factor of $1/2$ arises since each edge appears  twice in the matrix $A$.) For unweighted graphs, $A_{ij} =A_{ji}=1$ for each edge $(i,j)$ so that $M$ equals the number of edges. 
We further define  $D=\mbox{diag}[d_1,\dots,d_N]$, where each $d_i = \sum_j A_{ij}$ encodes the weighted degree (i.e., strength) of node $i$. The unnormalized  Laplacian matrix {(also called \emph{combinatorial Laplacian})}  is given by a size-$N$ square matrix
\begin{equation} \label{eq:Laplacian}
L = D - A.
\end{equation}
The matrix $L$ is important to many diverse applications including graph partitioning \cite{fiedler1989laplacian}, analysis of spanning trees  \cite{maurer1976matrix}, synchronization of nonlinear dynamical systems \cite{pecora1998master,skardal2014optimal,taylor2016synchronization}, electricity flow \cite{doyle1984random}
manifold learning \cite{belkin2001laplacian,coifman2006diffusion},  harmonic analysis \cite{coifman2006harmonic}, graph sparsification \cite{spielman2011graph}, and neural networks \cite{bronstein2017geometric,zhu2021contrastive}. 
Of particular relevance, one can define the diffusion equation over a graph by 
\begin{equation} \label{eq:heat}
\frac{d}{dt} {\bf x}(t) = -\hat{\beta} L {\bf x}(t),
\end{equation}
where ${\bf x}(t)=[x_1(t),\dots,x_N(t)]$, ${ x}_i(t)$ is the density at node $i$ at time $t$, and  $\hat{\beta}>0$ is a diffusion rate.

Theory development for many Laplacian-matrix-related applications   often stems from studying the spectral properties (i.e., eigenvalues and eigenvectors) of $L$.
As a real symmetric matrix, $L$ is diagonalizable,  $L=U\Lambda U^T$, where   $\Lambda = \mbox{diag}[\lambda_{1},\dots,\lambda_{N}]$ is a diagonal matrix whose diagonal entries encode the eigenvalues $0=\lambda_1 < \lambda_2\le \dots \le\lambda_{N}$, and their corresponding eigenvectors, ${\bf u}^{(k)}$, make up the orthonormal columns of  $U=[{\bf u}^{(1)},\dots,{\bf u}^{(N)}]$. 
We assume that the network is connected (i.e., $L$ is irreducible) so that $\lambda_2>0$.
By construction, $L{\bf 1}=0$, where ${\bf 1}=[1,\dots,1]^T$, which implies $\lambda_1=0$ is an eigenvalue and   ${\bf u}^{(1)}= N^{-1/2}{\bf 1}$ is its associated normalized eigenvector. To  facilitate later analyses, we also define the vector of eigenvalues: $\vec{\lambda} = [\lambda_{1},\dots,\lambda_{N}]$.  
%
Using a spectral decomposition, Eq.~\eqref{eq:heat}  has the general  solution 
\begin{equation} \label{eq:heat_sol}
{\bf x}(t) = e^{-\hat{\beta} t L}{\bf x}(0) =
Ue^{-\hat{\beta}  t \Lambda}U^{T}{\bf x}(0)
= \sum_{k=1}^N \alpha_{k} e^{-  \hat{\beta} t \lambda_k}  {\bf u}^{(k)}, 
\end{equation}
where $\alpha_{k} =   \langle {\bf u}^{(k)} ,{\bf x}(0)\rangle$ is the projection of the initial condition onto the $k$th eigenvector (each of which is an invariant subspace of the diffusion dynamics). Herein, we will often combine the diffusion rate and time parameter into a single \emph{timescale parameter} $\beta=\hat{\beta}t$.

Before continuing, we highlight that 
there exist   normalized versions of Laplacians that are also widely studied.  For example, $\tilde{L} = D^{-1}L$ is an asymmetric normalized Laplacian, and under the substitution $\tilde{L}\mapsto L$, Eq.~\eqref{eq:heat_sol} would describe a continuous-time Markov chain \cite{freedman2012markov}. Similarly, $\hat{L} = D^{-1/2}LD^{-1/2}$ defines a symmetric normalized Laplacian that shares the same eigenvalues as $\tilde{L}$ and is broadly important for machine learning applications. In this work, we will focus on $L$, noting that our  methodology and findings can be easily extended to applications involving these other matrices.

\subsection{Information diffusion and von Neumann entropy (VNE)}\label{sec:VNE}
While the diffusion equation given by Eq.~\eqref{eq:heat} is often associated with heat flow, it can also be used as a model for other types of  diffusion including the propagation of quanta of information over complex networks \cite{ghavasieh2020enhancing, ghavasieh2021unraveling, benigni2021persistence}. 
We begin by using $\beta=\hat{\beta}t$ to {define 
$Z_\beta( L) \equiv \text{Trace}(e^{-\beta{L}}) = \sum^N_{k = 1} e^{-\beta\lambda_k},$
%
%
%
which} normalizes a \emph{propagator matrix} $e^{-\beta{L}}$   
to yield an associated \emph{density matrix}
\begin{align} \label{eq:expo22}
\rho_\beta( L) &=  \frac{e^{-\beta L}} 
{Z_\beta( L)} 
= \sum_{k=1}^N s_k  {({\bf u}^{(k)})^T({\bf u}^{(k)})},
\end{align}
where 
{$s_k  = {e^{-\beta \lambda_k}}/{Z_\beta( L )}.$}
These expressions use the eigenvalues $\lambda_k$ and eigenvectors ${\bf u}^{(k)}$ of $L$, which were defined in the previous section.

We note that in the classical interpretation of diffusion, if a particle is first observed at node $i$ and later observed at node $j$, then it is believed to have taken some well-defined (although possibly unknown) path from $i$ to $j$. In contrast,  the quantum transport of a quantum particle that is emitted at node $i$ is considered to be de-localized (i.e., exists across all paths) until it is later detected at   node $j$.  
In either case, one can 
interpret a vector ${{\bf x}} (t) $ to be an \emph{information field} at timescale $t\ge 0$ for information dynamics that are initialized with an initial field ${{\bf x}} (0) =  \sum_{i=1}^N p_i{{\bf e}_i}$, where $p_i$ is the probability that information is ``seeded'' at node $i$ and ${{\bf e}_i}$ is a unit vector in which all entries are zeros except for the $i$th entry, which is a one.  We assume $\sum_i p_i=1$.
Each entry   $[\rho_\beta(  L)]_{ij}$ gives the probability of transport of information quanta from node $i$ to $j$ at  timescale $\beta$. 
Moreover, each eigenspace-restricted transport operator $s_k  {({\bf u}^{(k)})^T({\bf u}^{(k)})}$
defines a $k$th  \emph{information stream} over which an information quanta can be transported. By construction, $\text{Trace}(\rho_\beta( L)) = \sum_k s_k  = 1$ so that each $s_k$ defines the probability of transport along the $k$th information stream at timescale $\beta$. 
For both classical and quantum diffusion it is natural to examine whether transport is localized onto a small set of information streams (i.e., a few eigenmodes) or whether it is dispersed across many information streams (i.e., many eigenmodes). Such an inquiry is said to examine the \emph{spectral complexity} of information diffusion.
Here, we focus on classical diffusion according to Eq.~\eqref{eq:heat}  but seek to utilize the toolset of quantum information theory to measure the spectral complexity  of information diffusion, and we will later use it to measure the importance (i.e., centrality) of edges.
To this end, we quantify the spectral complexity of information diffusion using von Neumann entropy (VNE).
VNE was introduced by John von Neumann as a measure for quantum information \cite{neumann2013mathematische} and can quantify, for example, the departure of a quantum-mechanical system from its pure state. Applications include  using VNE to identify entangled spin-orbital bound states \cite{you2012orbitals}, examining non-equilibrium thermodynamics of bosons \cite{mendes2020bosons}, and developing complementary analyses for quantum phenomena \cite{holik2015natural}. Due to its widespread applicability and inherent ability to capture uncertainty for entangled quantum states, VNE has become a cornerstone for modern quantum information theory \cite{wilde2013quantum}.

Recently, this formalism was extended to study structural information for graphs   \cite{braunstein2006laplacian,de2015structural,li2018network}, although its more reliable generalization that preserves the   sub-additivity property of VNE, was proposed later in \cite{de2016spectral} and further developed in \cite{ghavasieh2020statistical}. (See \cite{ghavasieh2022statistical} for a review). Applications have included the study of global trade networks \cite{de2015structural} and the functional connectivity of brain networks \cite{nicolini2019brain,benigni2021persistence}.  We will utilize the sub-additive-preserving definition of VNE to develop an entropy-based centrality measure for edges in graphs.

We begin be defining a generalized notion of VNE for networks that incorporates and extends earlier versions.

\begin{definition}[von Neumann Entropy for Graphs]\label{def:VNE}
Let {$G(\mathcal{V},\mathcal{E})$} be a graph and $L$ be an associated  Laplacian matrix.
Assume $L$  has the diagonalization  $L=U\Lambda U^{-1}$, where  $\Lambda = \text{diag}(\vec{\lambda})$, $\vec{\lambda} = [\lambda_1,\dots,\lambda_N]^T$ with $\lambda_i\in\mathbb{R}$, and $U$ contains the associated (right) eigenvectors as columns.
Further, let $\rho(L)$ be a matrix-valued analytic function that admits the diagonalization  $\rho(L) = U \mbox{diag}[f_1(\vec{\lambda}),\dots,f_N(\vec{\lambda})] U^{-1}$ and satisfies $1 = \Tr(\rho(L)) = \sum_i f_i(\vec{\lambda})$.
We then define a general form for VNE of graph {$G(\mathcal{V},\mathcal{E})$} by
\begin{align} \label{eq:vne1}
h(L) &= -\Tr[ \rho(L) \log_2 \rho(L)] \nonumber\\
 &= - \sum_{k = 1}^{N} f_k(\vec{\lambda}) \log_2 f_k(\vec{\lambda}).
 \end{align}
By convention, we define $0 \log_2(0) = 0$. 
\end{definition}

\begin{remark}
Because $1 = \sum_k f_k(\vec{\lambda})$, one can interpret each $f_k(\vec{\lambda})$ as a probability, and then VNE coincides with Shannon entropy for the set    $\{f_k(\vec{\lambda})\}$ of probabilities. 
\end{remark}

\begin{remark}
We define a matrix-valued function  $\rho(L)$ that satisfies $\text{Trace}(\rho(L))=1$ 
to be a ``trace-normalized matrix transformation''.  In principle, matrix $L$ need not be restricted to  Laplacian matrices, and it could represent any  graph-encoding matrix that is diagonalizable and has  real-valued eigenvalues. 
\end{remark}

The earliest application of VNE to graphs that we know of can be attributed to Braunstein et al.\  \cite{braunstein2006laplacian}, wherein the authors restricted their attention to the choice of function 
 $f_k(\vec{\lambda}) = \lambda_k  /\sum_j \lambda_j$ for possibly weighted graphs, and we  will refer to the associated VNE as  \emph{uniform VNE}. For unweighted graphs,  $\sum_j \lambda_j=2M$ gives twice the number of edges. However, it has been shown that this definition does not preserve the desired property of sub-additivity for   entropies \cite{de2015structural}. 
 More recently, De Domenico and Biamonte  \cite{de2016spectral} proposed using $\rho(L)=\rho_\beta(L)$, as defined in Eq.~\eqref{eq:expo22}, in which case  $f_k(\vec{\lambda}) =  s_k$, as defined {for  Eq.~\eqref{eq:expo22}}. 
%
%
%
Importantly, this definition does satisfy the sub-additivity property, {which states that $h(L^{(1)} + L^{(2)}) \le h(L^{(1)}) + h(L^{(2)})$ for two graphs with equally-sized Laplacian matrices $L^{(1)},L^{(2)}\in\mathbb{R}^{N\times N}$}.


Our definition of VNE above allows  $L$ to represent either a normalized or unnormalized Laplacian matrix. 
%
Here, we focus on the choice of 
%
%
the combinatorial Laplacian matrix $L$    given by  Eq.~\eqref{eq:Laplacian}, and we refer to the corresponding VNE measure  as \emph{diffusion-kernel VNE}.
We note in passing, however, that   VNE has also been studied  using a left normalized  (also called random-walk) Laplacian matrix $\mathcal{L} = D^{-1}L $  \cite{ghavasieh2020enhancing}, and we will refer to that form    as \emph{continuous-time random-walk VNE.}
%
Regardless of the choice of matrix $L$, large VNE values indicate that  diffusion transport is dispersed across many   information streams (i.e., many eigensubspaces), whereas  small VNE values indicate that  it is localized to one, or a few,  information streams (i.e., concentrated onto a smaller eigensubspace). See \cite{ghavasieh2020statistical,benigni2021persistence,ghavasieh2022statistical} for the exploration of other relationships between VNE and information diffusion dynamics, including discussions on trapped fields and information field diversity.


\subsection{Spectral Perturbation Theory for Centrality Analysis}\label{sec:pert0}
%
Given that diffusion-kernel VNE measures the spectral complexity of information diffusion, the remainder of this paper will focus on VNE-based analyses of networks.
Specifically, we propose to rank edges (although our methods   easily extend to nodes and subgraphs) according to how each edge's  removal would change the graph's VNE. 

Our work is largely motivated by prior centrality measures that quantify importance by considering how structural modifications lead to perturbations for the spectral properties of graph-encoding matrices. 
Of particular relevance is recent work  \cite{ghavasieh2021unraveling} that proposed a node ranking  by examining how VNE changes occur after node removals, and  \emph{node entanglement} was defined as a centrality measure inspired by the concept of quantum entanglement for quantum-mechanical systems. 
Herein, we will develop a complementary edge centrality measure that ranks edges based on  how each edge's  removal would change the graph's VNE.
%
Importantly, centrality measures relating to spectral changes can be computationally expensive if many eigenvalues and eigenvectors are involved in the measures' definitions. In particular, VNE requires knowledge of all the eigenvalues, and so the direct recomputation of VNE for the Laplacian matrices of graphs after the removal of nodes or edges would be impractical for large graphs.

Given this limitation,  herein we propose  to build on existing techniques that leverage spectral perturbation theory to   approximate centrality measures and rankings in a way that is more computationally efficient.
For example, the \emph{dynamical importance} \cite{restrepo2006} of a node or edge is defined as the decrease that would occur for the adjacency matrix's  spectral radius if that node or edge is removed (see also \cite{tong2012gelling}).  {In addition, one can define the importance of a node by computing how much the PageRank or other centrality measures change upon the  removal of that node \cite{du2008perturbationrank}.}  Other related work includes considering the  impact of structural modifications on the spectra of Laplacian matrices \cite{milanese2010approximating,li2018network,song2021asymmetric} as well as functions of such eigenvalues and eigenvectors \cite{taylor2016synchronization}. Such centrality measures can be efficiently approximated by developing linear approximations for how structural perturbations of graphs give rise to spectral perturbations. We also refer the reader to \cite{taylor2017eigenvector,taylor2019tunable} for the use of more advanced matrix-perturbation techniques (namely singular perturbation theory) that  have arisen in the context of centrality measures for multiplex and temporal networks. We note in passing that these perturbative  analyses of networked-coupled dynamical systems are closely  related to the notions of eigenvalue and eigenvector elasticities \cite{kampmann1996feedback,kampmann2006loop}.

Spectral perturbation theory for network  modifications often stems from the following first-order approximation.

\begin{theorem}[Perturbation of Simple Eigenvalues  \cite{atkinson2008introduction}] \label{def:pert_thm}
Let $ {X}$ be a symmetric $N\times N$ matrix with eigenvalues $\{\lambda_j\}$ and normalized eigenvectors $\{{\bf u}^{(j)}\}$. Consider an  eigenvalue $\lambda_i$   that is simple in that it has algebraic  multiplicity one: $\lambda_i\not=\lambda_j$ for any other $j$ and $\lambda_i$ has a 1-dimensional eigenspace spanned by the eigenvector  ${\bf u}^{(i)}\in \mathbb{R}^N  $. 
Further, consider a fixed symmetric perturbation matrix $\Delta  {X}$, and let $ {X}(\epsilon)= {X}+\epsilon \Delta  {X}$.
We denote the eigenvalues  {and eigenvectors} of $ {X}(\epsilon)$ by $\lambda_j(\epsilon)$  {and ${\bf u}^{(j)}(\epsilon)$, respectively.
It then follows that
\begin{align}
\lambda_i(\epsilon)  &= \lambda_i + \epsilon \lambda'_i(0) + \mathcal{O}(\epsilon^2),  
\label{eq:eigen_perturb_a}
\end{align}
where $\lambda_i'(0)$ is the derivative of $\lambda_i(\epsilon)$ with respect to $\epsilon$ at $\epsilon=0$ and is given by 
\begin{align}
\lambda_i'(0) & = ({\bf u}^{(i)})^T \Delta  {X} {\bf u}^{(i)}  .
\label{eq:eigen_perturb_b}
\end{align}
{Here, $\mathcal{O}(\epsilon^2)$ indicates that the difference between the right and left-hand sides has an asymptotic $\epsilon \to 0$ scaling behavior that is {upper} bound by {a constant multiple of} $\epsilon^2$.}
}
\end{theorem}
\begin{proof}
See \cite{atkinson2008introduction}.
\end{proof}
\begin{remark}
The first-order approximation given by Eq.~\eqref{eq:eigen_perturb_a} is accurate when the perturbation is sufficiently small,
$|\epsilon\lambda'_i(0)|\ll |\lambda_i|$.
\end{remark}

To study how edge perturbations change the spectral properties of a   Laplacian matrix $L$ given by Eq.~\eqref{eq:Laplacian}, we can set $X=L$ and study a perturbation  $\Delta L$ encoding the changed edge(s).

\begin{proposition}[Weighted Edge Perturbations {for Laplacian matrices}] \label{cor:1edge}
When the weighted and undirected network is modified by adding an edge $(p,q)$ of weight $A_{pq}$, the Laplacian perturbation matrix takes the form
\begin{align}
\Delta  {L}_{ij}^{(pq)} = \left\{ \begin{array}{rl} 
A_{pq} ,&  {(i,j)\in\{(p,p),(q,q)\}}\\ 
-A_{pq} ,&  {(i,j)\in\{(p,q),(q,p)\}}\\ 
0,&\text{otherwise}. 
\end{array}\right.
\label{eq:DL}
\end{align}
Similarly, when the network is modified by removing an edge $(p,q)$, the corresponding Laplacian perturbation matrix is $-\Delta  {L}_{ij}^{(pq)}$. (We emphasize that in this notation we use $A_{pq}$ to be the (nonzero) weight of the edge when present --- that is, the weight of the newly added edge or the weight of the edge prior to its removal.)
\end{proposition}
Versions of Proposition~\ref{cor:1edge} have appeared in various literatures, usually under the assumption of unweighted networks (see, e.g., \cite{taylor2016synchronization,li2018network}, in which case $A_{ij}\in\{0,1\}$). 


{This form} for $\Delta L$ can be combined with Thm.~\ref{def:pert_thm} to provide first-order approximations for the impact of network modifications on the eigenvalues of a Laplacian matrix.

\begin{lemma}[First-order Spectral Impact of Edge Additions and Removals \cite{milanese2010approximating}] \label{cor:spec_pert}
Let $\Delta  {L}^{(pq)}$ and $-\Delta  {L}^{(pq)} $, respectively, denote the perturbation matrix for a graph Laplacian under the addition and removal of a weighted edge $(p,q)$ with weight $A_{pq}$ as defined in Prop.~\ref{cor:1edge}. Then the simple eigenvalues $\lambda_i(\epsilon)$ of the perturbed matrix $L\pm \epsilon \Delta L^{(pq)}$ are given by Eq.~\eqref{eq:eigen_perturb_a} with
\begin{align}
\lambda_i'(0) = \pm ({\bf u}^{(i)})^T \Delta  {L}^{(pq)} {\bf u}^{(i)} &= \pm A_{pq}({ u}^{(i)}_p-{ u}^{(i)}_q)^2.
\end{align}
\end{lemma}
%


In the next section, we extend this spectral approximation theory to diffusion-kernel VNE to develop VNE-based edge-ranking algorithms that are computationally efficient and can scale to large networks.
It is also worth noting that the above-mentioned spectral perturbation theory has been previously used to study graph VNE \cite{li2018network}, although that work {did not investigate centrality and} focused on a formulation in which {$s_k = \lambda_k / 2M$}.
{Finally, we note that the above results 
easily extend to describe the first-order effects due to the modification of sets of edges  \cite{taylor2016synchronization}, thereby 
allowing predictions for node and subgraph modifications (although that is not our focus here).}
%

\section{VNE perturbations measure the   importance of edges}\label{sec:results}
%
We now present our main theoretical and algorithmic results.
Specifically, we introduce   algorithms that rank edges according to  their contribution to the spectral complexity of information diffusion over networks, as measured by diffusion-kernel VNE.
In Sec.~\ref{sec:Rank}, we develop these rankings and discuss their limitation to small graph sizes.
In Sec.~\ref{sec:Rank2}, we develop an approximate-ranking algorithm that efficiently scales to large graphs and is based on spectral perturbation theory.
In Sec.~\ref{sec:runtime}, we present experiments to compare these two rankings, which are based on the true and approximate values for how VNE would be perturbed by edge removals. (Further experiments are deferred to Sec.~\ref{sec:emp_exp}.)

\subsection{Edge rankings by VNE increases upon  edge removals}\label{sec:Rank}

Given a graph $G(\mathcal{V},\mathcal{E})$ with $N$ nodes and $|\mathcal{E}|$ undirected edges with weights $\{A_{ij}\}$, we seek to rank the edges so that the most important edge $(p,q)\in\mathcal{E}$ is the one that would maximize the VNE of a residual graph that results from removing edge $(p,q)$. 
More precisely, let $L$ be the unnormalized graph Laplacian of $G(\mathcal{V},\mathcal{E})$ given by Eq.~\eqref{eq:Laplacian} and  $-\Delta L^{(pq)}$ be the change to $L$ given by Eq.~\eqref{eq:DL} that would occur upon the removal of edge $(p,q)\in\mathcal{E}$. Further, let $h(L)$ be the graph's VNE, as given by Eq.~\eqref{eq:vne1}. Then   the change to VNE that would occur upon the removal of edge $(p,q)$ is given by
\begin{align}\label{eq:VNE_change}
Q_{pq} = h(L-\Delta L^{(pq)}) - h(L).
\end{align}
In principle, $Q_{pq}$ is not necessarily positive, although our experiments suggest that $Q_{pq}>0$ for most edges.
Considering the set $\{Q_{pq}\}$ of VNE perturbations for $(p.q)\in\mathcal{E}$, the top-ranked edge is a solution to the following optimization problem
\begin{equation}
    \label{eq:Qactual}
    (p,q) = \text{argmax}_{(p,q)\in\mathcal{E}}~Q_{pq}.
\end{equation}

\begin{definition}[Edge Rankings by VNE Increases]
Given a graph $ {G}=(\mathcal{V},\mathcal{E})$ with unnormalized Laplacian $L$, consider the changes $Q_{pq}$ to VNE  given by Eq.~\eqref{eq:VNE_change}, which  would occur upon the removals  of  edges $(p,q)\in\mathcal{E}$. Then, we define the rankings
\begin{align}
R_{pq} = 1+ | \hat{\mathcal{E}}|,~\text{where}~ \hat{\mathcal{E}} = \{ (n,m)\in\mathcal{E} : Q_{nm} > Q_{pq}\}\label{eq:rank1}
\end{align}
so that $R_{pq}\in\{1,\dots,|\mathcal{E}|\}$ gives the rank of each edge $(p,q)\in\mathcal{E}$. 
\end{definition}
{In practice, we implement these rankings efficiently using the NumPy function `argsort` in Python.}

\begin{algorithm}[h]
\caption{{\bf Edge Rankings by VNE Increases}}
\begin{algorithmic}[1]\label{alg:actual}
\REQUIRE Graph $G(\mathcal{V},\mathcal{E})$ with nodes $\mathcal{V}$, weighted edges $\mathcal{E}$, and unnormalized Laplacian $L$.
\ENSURE Cardinality $|\mathcal{E}|>0$ 
\STATE Compute set $\{\lambda_i\}$ of eigenvalues for $L$.
\STATE Compute VNE $h(L)$ according to Eq.~\eqref{eq:vne1}.
\FOR{$(p,q)\in \mathcal{E}$}
\STATE {Compute the perturbation matrix $\Delta L^{(pq)}$ according to Eq.~\eqref{eq:DL}}.
\STATE Compute the change $Q_{pq}$ according to Eq.~\eqref{eq:VNE_change} {using $L$ and $\Delta L^{(pq)}$}.
\ENDFOR
\STATE {Use Eq.~\eqref{eq:rank1}} on the set $\{{Q}_{pq}\}$ to obtain the set  $\{R_{pq}\}$ of rankings.
\end{algorithmic}
\end{algorithm}

\begin{remark}
One can rank edges according to  perturbations to VNE given by Eq.~\eqref{eq:vne1} for any choice of the function $\rho(\lambda)$. We will focus herein on the diffusion-kernel VNE in which $\rho(\lambda)$ is given by {Eq.~\eqref{eq:expo22}}.
\end{remark}

\begin{remark}
Note that the $Q_{pq}$ values for different edges $(p,q)$ can be the same, in which case the rankings are not   unique. Such scenarios can be handled in a variety of ways. We resolve edges with tied rankings {arbitrarily}. That is, if 
$k$ edges are tied for rank $R$, then we {arbitrarily} assign them ranks $R$, $R+1$, \dots, $R+k$.
\end{remark}

Before continuing, we highlight that the direct computation of rankings $R_{pq}$ will be infeasible for large graphs due to the large computational cost. For example,   in   Sec.~\ref{sec:runtime} we estimate that the runtime to compute the set $\{Q_{pq}\}$ for a graph with  $N=8000$ nodes would take approximately 1300 hours (about 54 days) to complete on a standard desktop. 
{More specifically}, computing all  eigenvalues $\{\lambda_i\}$ of a {size-$N$} matrix $L$ is an expensive task 
{having} computational complexity {that typically scales proportionally to $N^3$ (e.g., if using the QR method \cite{golub2013matrix})}. In order to compute the VNE change $Q_{pq}$ given by Eq.~\eqref{eq:VNE_change} for an edge $(p,q)$, one needs to first compute the perturbed eigenvalues $\{\lambda_i'\}$ for  a perturbed Laplacian $L'=L - \Delta L^{(p,q)}$. Since this must be done for each edge, constructing the set $\{Q_{pq}\}$ for $(p,q)\in\mathcal{E}$ has computational complexity {that is expected to scale proportionally to $|\mathcal{E}|N^3$. That said, developing numerical methods to compute the eigenvalues for large sparse matrices remains an active research field, and in  Sec.~\ref{sec:Discuss} we survey some leading approaches  that may be useful to speed up Algorithm~\ref{alg:actual}.} 


{Nevertheless, it is impractical for large networks to numerically compute the eigenvalues for $|\mathcal{E}|$ distinct matrices   of size $N$.}
 Thus motivated,  in the next section we develop an approximate ranking algorithm that is much more efficient  and can be readily applied to larger graphs.  

\subsection{Perturbation theory for VNE  and efficient edge ranking}\label{sec:Rank2}
%
Implementing Algorithm~\ref{alg:actual} can be computationally infeasible for large graphs, and here we develop approximate rankings based on  the spectral perturbation theory that we presented in Sec.~\ref{sec:pert0}.
Recall that we define VNE (see Def.~\ref{def:VNE}) using a spectral map
%
%
%
${\bf f}:\mathbb{R}^N\mapsto \mathbb{R}^N$ that is defined entry-wise by
\begin{equation} \label{eq:e2}
{\bf f}(\vec{\lambda})=[{  f}_1(\vec{\lambda}),\dots,{  f}_N(\vec{\lambda})]^T,
\end{equation}
where  $\vec{\lambda}=[\lambda_1,\dots,\lambda_N]^T$ is a vector of eigenvalues. For uniform VNE, one has $f_k(\vec{\lambda}) = \lambda_k/(2M)$, which depends only on the $k$th eigenvalue since $\sum_k \lambda_k = 2M$ is a conserved quantity equal to total edge weight.
In contrast, for diffusion-kernel VNE  one has that  $ f_k(\vec{\lambda}) = s_k  = e^{-\beta\lambda_k}/Z(\beta)$, as defined  {for  Eq.~\eqref{eq:expo22}}. 
%
In this latter case, each $f_k(\vec{\lambda})$ depends on all the eigenvalues due to the denominator term $Z(\beta)=\sum_k e^{-\beta\lambda_k}$.

We now  approximate  how structural modifications change VNE to first order.

\begin{lemma}[First-Order Perturbation of VNE]\label{def:gen_FOP_0}
Let $L$ be a diagonalizable Laplacian matrix  given by Eq.~\eqref{eq:Laplacian} for an undirected graph and $L' = L +\epsilon \Delta L$ denote its perturbation for a small scalar $\epsilon$. Further, let $h(L)$ be the graph's VNE given by Def.~\ref{def:VNE}.
Letting $H(\epsilon)\equiv h(L+\epsilon \Delta L) $ be the $\epsilon$-perturbed VNE, then    $H(\epsilon)$   has the   first-order approximation
\begin{equation} \label{eq:taylor_series}
H(\epsilon) = H(0) + \epsilon H'(0) + \mathcal{O}(\epsilon^{2}),
\end{equation}
where $H(0)=h(L)$ is the VNE of the original graph.
Morever, when the eigenvalues $\lambda_i$ of $L$ are simple, then  the derivative
\begin{equation}\label{eq:partials}
    H'(0) = \left. \sum_{i,j} \dfrac{\partial h}{\partial f_i}\dfrac{\partial f_i}{\partial \lambda_j}  \dfrac{\partial \lambda_j}{\partial \epsilon} \right|_{\epsilon=0}\nonumber,
\end{equation}
can be constructed using  partial derivatives with  $\left.  {\partial \lambda_j}/{\partial \epsilon} \right|_{\epsilon=0} = \lambda'_j(0) =({\bf u}^{(j)})^T \Delta  {L} {\bf u}^{(j)}$  given by Eq.~\eqref{eq:eigen_perturb_b} in Thm.~\ref{def:pert_thm}. Here, ${\bf u}^{(j)}$ is the eigenvector associated with $\lambda_j$.
\end{lemma}

Lemma~\ref{def:gen_FOP_0} can be generally applied to different choices for the density function $\rho(L)$ described in Def.~\ref{def:VNE}. For example, it can recover a  previously shown result for uniform VNE.

\begin{theorem}[First-Order Perturbation of Uniform VNE \cite{li2018network}]\label{def:gen_FOP_0d}
Consider a uniform density function $f_i(\vec{\lambda}) =\rho(\lambda_i) = \lambda_i/\sum_j \lambda_j$ and the full set of assumptions in Lemma~\ref{def:gen_FOP_0}. Then the perturbed VNE is given by Eq.~\eqref{eq:taylor_series}  with
\begin{equation} \label{eq:vne_perturb3}
H'(0) = -\frac{1}{2M} \sum_{i }   ({\bf u}^{(i)})^T \Delta  {L} {\bf u}^{(i)}  \left[ \log_2 \left( \frac{\lambda_i}{\sum_j \lambda_j} \right) + \frac{1}{\ln(2)} \right] .
\end{equation}
\end{theorem}

In \cite{li2018network}, Thm.~\ref{def:gen_FOP_0d} was developed to study the effects of graph rewiring on uniform VNE, and in particular, how the distribution of VNE across a random-graph ensemble that is obtained via edge rewiring can converge to that for well-known random graph ensembles, including ensembles associated with the Erd\H{o}s-R\'enyi $G_{NM}$ and configuration models.     Turning our attention back to diffusion-kernel VNE, we obtain the following.

\begin{theorem}[First-Order Perturbation of diffusion-kernel VNE]\label{def:gen_FOP}
Consider the diffusion-kernel density function {$f_i(\vec{\lambda})=s_i$, as defined
according to Eq.~\eqref{eq:expo22},} and the full set of assumptions in Lemma~\ref{def:gen_FOP_0}.  Then the perturbed VNE is given by Eq.~\eqref{eq:taylor_series}  with 
\begin{align}\label{eq:perturbation}
H'(0) &= -\beta \sum_{i = 1}^{N} f_i(  \vec{\lambda}) \left[\log_2(f_i(  \vec{\lambda}))  + \frac{1}{\ln\left(2\right)}\right]\left[-\lambda_i'(0) +  \sum_{j}^{N}f_j( \vec{\lambda})\lambda_j'(0)  \right] ,
\end{align}
where $\lambda_i'(0) =  ({\bf u}^{(i)})^T \Delta  {L} {\bf u}^{(i)}$.
\begin{proof}
See Appendix ~\ref{app:proof_vne}.
\end{proof}
\end{theorem}

In principle, a Laplacian perturbation matrix $\Delta L$ can encode a broad family of structural modifications to a graph. For example, Lemma \ref{cor:spec_pert} in Sec.~\ref{sec:pert0}  showed that  $\lambda_i'(0) = \pm ({\bf u}^{(i)})^T \Delta  {L}^{(p,q)} {\bf u}^{(i)}  = \pm A_{pq}(u_p^{(i)}-u_q^{(i)})^2$ for the addition (+) or removal (-) of an edge $(p,q)$ having weight $A_{pq}$. 
{As shown in \cite{taylor2016synchronization}, a similar result can encode the simultaneous addition/removal for sets of edges.}
Thus, {in principle} one can compute Eq.~\eqref{eq:vne_perturb3} and Eq.~\eqref{eq:perturbation} very efficiently for a given edge modification or a  set of edge modifications. 
See Appendix~\ref{app:valid} for an experiment that provides numerical validation for Thm.~\ref{def:gen_FOP} by estimating the change to VNE induced by removing $k$ edges for $0\le k\le 20$.
%

With this approximation theory in hand, we now return our attention to the ranking of edges based on the change $Q_{pq} = h(L-\Delta L^{(pq)}) - h(L) $ to VNE that would occur after the individual removal of each edge $(p,q)\in\mathcal{E}$. Combining  Thm.~\ref{def:gen_FOP} with Lemma~\ref{cor:spec_pert}, we  approximate
%
\begin{align}\label{eq:app123}
  Q_{pq}\approx \tilde{Q}_{pq} \equiv H'(0),
\end{align}
where  $H'(0)$ is appropriately defined for the particular type of VNE (i.e., as represented by the spectral mapping $f_i(\vec{\lambda})$) and  $\lambda_j'(0) =  -A_{pq}(u_p^{(j)}-u_q^{(j)})^2 $.   Herein, we will focus on diffusion-kernel VNE in which $H'(0)$ is given by  Eq.~\eqref{eq:perturbation}, which uses   $f_i(\vec{\lambda})  =   {e^{-\beta\lambda_i}} /{\sum_{j  }  e^{-\beta\lambda_{j}}}$. In principle, one can use a similar approach to approximate edge rankings based on VNE perturbations for other choices of   $ f_i(\vec{\lambda})$. Finally, we also highlight that due to the linearity of first-order approximations, one can also estimate  the change to VNE that would occur due to the removal of a subset $\mathcal{E}^1\subset \mathcal{E}$ of edges  by
$H'(0) =  \sum_{(p,q)\in \mathcal{E}^1} Q_{pq} .$

Given the approximate perturbations $\{\tilde{Q}_{pq}\}$, we now define an approximate ranking of edges.

\begin{definition}[Edge Rankings by First-Order Approximate  VNE Increases]
Given a graph $ {G}=(\mathcal{V},\mathcal{E})$ with unnormalized Laplacian $L$, consider the approximate changes $\tilde{Q}_{pq}$ to VNE  given by Eq.~\eqref{eq:app123}.

Then,  then we define  the   rankings
\begin{align}
\tilde{R}_{pq} = 1+ | \hat{\mathcal{E}}|,~\text{where}~ \hat{\mathcal{E}} = \{ (n,m)\in\mathcal{E} : \tilde{Q}_{nm} > \tilde{Q}_{pq}\}\label{eq:rank1_approx}
\end{align}
so that $\tilde{R}_{pq}\in\{1,\dots,|\mathcal{E}|\}$ gives the approximate rank of each edge $(p,q)\in\mathcal{E}$. 
\end{definition}
We summarize the computation of this ranking in Algorithm~\ref{alg:approximate}.

\begin{algorithm}[t]
\caption{{\bf Edge Rankings by First-Order-Approximate VNE Increases}}
\begin{algorithmic}[1]\label{alg:approximate}
\REQUIRE Graph $G(\mathcal{V},\mathcal{E})$ with nodes $\mathcal{V}$, weighted edges $\mathcal{E}$, and unnormalized Laplacian $L$.
\ENSURE Cardinality $|\mathcal{E}|>0$ 
\STATE Compute set $\{\lambda_i\}$ of eigenvalues for $L$.
\STATE Compute VNE $h(L)$ according to Eq.~\eqref{eq:vne1}.
\FOR{$(p,q)\in \mathcal{E}$}
\STATE {Compute the perturbation matrix $\Delta L^{(pq)}$ according to Eq.~\eqref{eq:DL}}.
\STATE Compute the approximate change $\tilde{Q}_{pq}$ according to Eq.~\eqref{eq:app123} {using $L$ and $\Delta L^{(pq)}$}.
\ENDFOR
\STATE {Use Eq.~\eqref{eq:rank1_approx}} on the set $\{\tilde{Q}_{pq}\} $ to obtain the set  $\tilde{R}_{pq}\}$ of approximate rankings.
\end{algorithmic}
\end{algorithm}

Before continuing, we highlight that one possible limitation of this approach is that Lemma \ref{def:gen_FOP_0},   Thm.~\ref{def:gen_FOP_0d} and Thm. \ref{def:gen_FOP}
define $H'(0)$ in a way that assumes  that the  eigenvalues $\lambda_j$ of   $L$ are simple {by virtue of Thm.~\ref{def:pert_thm}}. We propose  to neglect the contributions of repeated eigenvalues as a simple heuristic.
Moreover, we predict that if eigenvalue repetition occurs, it is more likely to arise for   large eigenvalues (i.e., for which $e^{-\beta \lambda_j}$ is  small), as opposed to the  smaller, more important   eigenvalues for which the  $e^{-\beta \lambda_j}$ terms are larger. 
Therefore, we  expect that the neglection of $\lambda_i'(0)$ terms for repeated eigenvalues will have a particularly small effect on an estimate $H'(0)$ for diffusion-kernel VNE and the subsequent rankings of edges. 
{However, while Algorithm~\ref{alg:actual} does not make any assumption about repeated eigenvalues, Algorithm~\ref{alg:approximate} does make such an assumption which limits its usability.  That said, we have yet to encounter repeated eigenvalues in our experiments presented here, but such situations are known to occur for some networks.}

\subsection{Comparison of Algorithms~\ref{alg:actual} and  \ref{alg:approximate}}\label{sec:runtime}

Recall that our main motivation for developing approximate rankings was that it is can be computationally expensive (or infeasible) to directly compute the set $\{Q_{pq}\}$, which requires one to compute the exact change to VNE that occurs after removing each edge $(p,q)$. More specifically, upon each edge removal, we must recompute all eigenvalues of a size-$N$ matrix.
{Computing all the eigenvalues of a matrix is known to be computationally expensive, and such an algorithm must be implemented $M$ separate times for $M$ distinct matrices if one considers the removal of $M$ edges. The QR algorithm, for example, requires $\varpropto N^3$ flops \cite{golub2013matrix}, implying that the computation time grows proportionally to $N^3$ for a network with $N$ nodes\footnote{{Notably, we implemented our algorithms using the NumPy Python package, and the eigenpair computations are based on LAPACK \cite{harris2020array}, for which the asymptotic runtime complexity has a lower bound $\Omega(N^2)$ and upper bound $\mathcal{O}(N^3)$ \cite{demmel2008lapack}.  Computations were conducted on a Dell Precision 3650 computer that has a 2.80 GHz Intel Xeon W-1390 processor with 64 GB RAM.}}. 
(Of course, speeding the computation of spectra for large sparse matrices  remains an active research field, and we highlight several leading approaches in Sec.~\ref{sec:Discuss}.)
}
In contrast, the approximate rankings utilize spectral perturbation theory and the eigenvalues of a size-$N$ matrix only must be computed once. Then for each edge, one   computes $\tilde{Q}_{pq}=H'(0)$ using Eq.~\eqref{eq:perturbation}, which has $\mathcal{O}(N)$ complexity and is a much more efficient approach.  {We also note  for both algorithms that there is an additional start-up cost {from having} to compute the spectra of the original graph Laplacian.}

In Fig.~\ref{fig:runtime}, we compare Algorithms~\ref{alg:actual} and  \ref{alg:approximate} in two ways. First, we show in Fig.~\ref{fig:runtime}(A)  that the approximate rankings $\{\tilde{R}_{pq}\}$ are much more efficient to compute   than the true rankings $\{{R}_{pq}\}$ for an example graph. Second, we show in Fig.~\ref{fig:runtime}(B)   that these two sets of rankings can be very similar $\tilde{R}_{pq} \approx R_{pq} $. Their similarity   is somewhat expected, since  $\tilde{Q}_{pq}\approx Q_{pq}$ by construction.

\begin{figure}[h]
\centering
\includegraphics[width= \columnwidth]{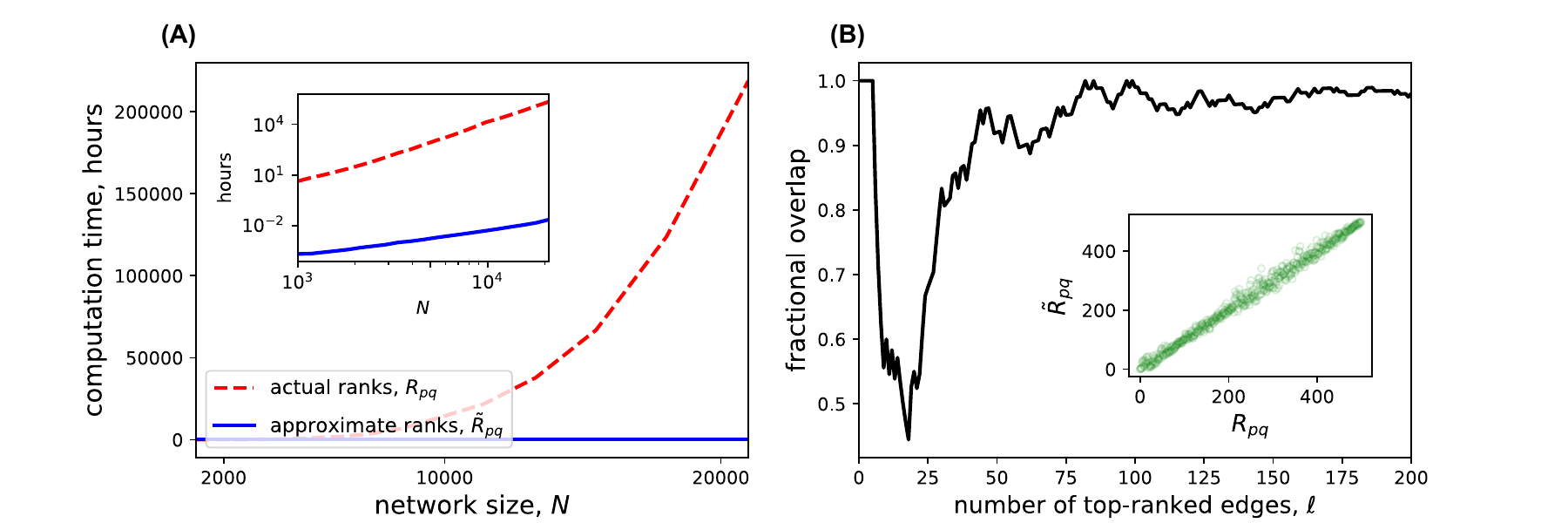}
\label{fig:runtime}
\vspace{-.5cm}
\caption{{\bf Comparison of Algorithms~\ref{alg:actual} and \ref{alg:approximate}.}
{\bf (A)} We estimate the runtimes for {$\beta = 1$} that are required to execute these algorithms for  Erd\H{o}s--R\'enyi $G_{NM}$ graphs with increasing size with  $N$ nodes and $M=5N$ edges, so that   mean degree remains fixed as $N$ increases.
The inset is a logarithmic scale to reveal their runtimes' scaling behaviors (see text). Note that  Algorithm~\ref{alg:actual} would potentially take hundreds of hours for each $N$, so we extrapolated these times from that required to compute for a single edge, instead of running the algorithm to completion. In contrast, Algorithm~\ref{alg:approximate}  completed in under a minute for all values of $N$ plotted.
{\bf (B)} We compare the fractional overlap $|\mathcal{E}^{(\ell)} \cap \tilde{\mathcal{E}}^{(\ell)}| / \ell$ for the $\ell$  top-ranked edges for the two rankings $R_{pq}$ and $\tilde{R}_{pq}$ in a $G_{NM}$ graph with $N=100$ nodes and $M=1000$ edges {at $\beta = 1$}. The inset shows that they are strongly correlated: $R_{pq}\approx \tilde{R}_{pq}$ for all edges $(p,q)\in\mathcal{E}$. We computed a  Pearson correlation coefficient of $0.996$.
}
\end{figure}

In Fig.~\ref{fig:runtime}(A), we support our predictions for the computational complexity for both algorithms by studying the scaling behavior for the algorithms' runtimes for graphs of increasing size $N$.  We considered Erd\H{o}s--{R\'enyi} $G_{NM}$ graphs with   $N$ nodes and $M=N\langle d\rangle/2$ edges, where $ \langle d\rangle=10$ is   the mean degree. We then generated a sequence of graphs    with increasing size $N$, and for each graph we empirically estimated the runtime required to compute the rankings $R_{pq}$ and $\tilde{R}_{pq}$. For the approximate rankings $\{\tilde{R}_{pq}\}$, we computed them directly using Algorithm~\ref{alg:approximate}. However, we found that it was computationally infeasible to directly run Algorithm~\ref{alg:actual} for these graphs with sizes $N>2000$, and so we instead estimated the algorithms' runtime by measuring the time required to  compute $Q_{pq}$ for a single edge. Then we approximated the total run time by multiplying this duration by $M$, since the algorithm requires us to compute $Q_{pq}$ for all edges.
We considered $N\in[2,000,20,000]$ and found that in all cases Algorithm~\ref{alg:approximate} completes in less than a minute, whereas Algorithm~\ref{alg:actual} would take hundreds of hours to perform all of these computations (i.e., were we to attempt running it in its entirety).
The subpanel in Fig.~\ref{fig:runtime}(A) showing the runtimes in a log-log scale, and we applied least-squares linear fits to empirically estimate the {scaling to be $\propto MN^{2.813}$} for Algorithm~\ref{alg:actual} and {$\propto MN^{0.519}$} for Algorithm~\ref{alg:approximate}.  
{We note that these observing scaling behaviors are {slightly} less than {our} predicted 
scaling behaviors:
 { $\varpropto  MN^3$ } for Algorithm~\ref{alg:actual} and {$\varpropto MN $} for Algorithm~\ref{alg:approximate}.
%
%
Notably,  our experiments were limited to graphs with  $N = 20,000$ nodes, and we {expect} that the consideration of larger graphs would yield observed scaling rates closer to these {asymptotic} predictions.
}

In Fig.~\ref{fig:runtime}(B), we provide evidence to support our claim that the approximate and  true rankings are similar: $\tilde{R}_{pq} \approx R_{pq} $.  Here, we computed the rankings ${R}_{pq} $ and $ \tilde{R}_{pq} $ for an Erd\H{o}s--R\'enyi $G_{NM}$ graph {\cite{erdos1959random}} with  $N = 100$ nodes and $M=1000$ edges. For each set of rankings, we considered $\ell$ top-ranked edges and measured the \emph{fractional overlap}, {also known as the \emph{precision at $\ell$} metric \cite{clough2013evaluating}}, between these two sets, or more precisely, the size of their overlap. That is, letting ${\mathcal{E}}^{(\ell)} $ and $\tilde{\mathcal{E}}^{(\ell)}$ be the sets of top-ranked edges according to $R_{pq}$ and $\tilde{R}_{pq}$, respectively, we plot $ {|\mathcal{E}^{(\ell)} \cap \tilde{\mathcal{E}}^{(\ell)}|}/{\ell}$ versus the set size $\ell$.
Note that this is a fraction that lies between 0 and 1, since   the sets of top-ranked edges have the same cardinality: $\ell = |{\mathcal{E}}^{(\ell)}|=|\tilde{\mathcal{E}}^{(\ell)}|$. Observe in Fig.~\ref{fig:runtime}(B) that this fraction is very large for many values of $\ell$, implying that the rankings $\{R_{pq}\}$ and $\{\tilde{R}_{pq}\}$ are very similar for this graph. This is further shown in the subplot, where we provide a scatter plot that compares $\tilde{R}_{pq}$ versus ${R}_{pq}$ for each edge $(p,q)\in\mathcal{E}$.

{In Appendix~\ref{app:valid2}, we extend and recapitulate the experimental findings shown in Fig.~\ref{fig:runtime}(B) by obtaining similar results for two empirical networks that will be later described in Sec.~\ref{sec:emp_exp}.}
Given the {strong} similarity between rankings $\{R_{pq}\}$ and $\{\tilde{R}_{pq}\}$, and the observation that computing $\{R_{pq}\}$ directly using Algorithm~\ref{alg:actual} is inefficient (or infeasible) for large graphs,  we will focus our attention to studying the approximate rankings $\{\tilde{R}_{pq}\}$ given by Algorithm \ref{alg:approximate} for the remainder of this paper.

\section{Three empirical case studies reveal structural/dynamical mechanisms that influence the   importance of edges}\label{sec:emp_exp}
%
Here,  we apply our VNE-based measures for edge importance to three empirical network datasets:
(Sec.~\ref{sec:congress_network}) a voting-similarity network for the U.S. Senate, 
(Sec.~\ref{sec:transit_network}) a multimodal transportation system, and 
(Sec.~\ref{sec:brain_network}) a multiplex brain network. 
In each empirical network, we will show that the edges that are deemed most important   can drastically change by considering different   timescale parameters $\beta$.
In so doing, we will examine the crucial role that is played by  the interplay between $\beta$
and some network property:  
(Sec.~\ref{sec:congress_network}) community structure due to political polarization, 
(Sec.~\ref{sec:transit_network})   edge weights that encode different speeds  along roads and metro lines, and
(Sec.~\ref{sec:brain_network}) the strength of coupling between   layers of a multiplex  network. 
{In Sec.~\ref{app:centrality_comparison}, we compare our proposed centrality measure to other edge centralities for these three empirical case studies.}
%

\subsection{U.S. Senate voting-similarity network with community structure}\label{sec:congress_network}

In our first experiment, we study the importance of edges according to Algorithm~\ref{alg:approximate} for a network that encodes voting similarity among U.S. Senators, and  we will compare the rankings  of \emph{interparty edges} that connect Senators in different parties to the rankings of \emph{intraparty edges} between Senators in the same party.  More generally, we ask the following: which edges are most important in a graph that contains community structure, the edges between communities or the ones inside of communities? As we shall show, the answer will depend sensitively  on  the timescale of the dynamics considered.

For brevity, we defer a detailed description of the voting-similarity network that we study to   Appendix~\ref{app:con_weights}. Here, we provide a summary. We created a Python Module called \emph{VoteView-python} \cite{voteview_code} that extracts data from the VoteView repository \cite{lewis2018voteview} that summarizes the voting records of members of the U.S. Congress. 
Following techniques similar to those in \cite{waugh2009party,mucha2010communities,mucha2013polarization}, we constructed a graph with an associated adjacency matrix such that each entry $A_{ij}$ encodes the fraction of bills in which Senators $i$ and $j$ vote identically. We restrict our attention to  the 117th U.S. Senate considering bills from January 3, 2021 until June 30, 2022 (noting that the 117th Congress had not yet concluded when we conducted this experiment.) We also thresholded the matrix, setting $A_{ij}$ to zero for any fraction less than $0.4$. The resulting procedure yielded a graph with $N=100$ nodes and $|\mathcal{E}|=2,656$ undirected, weighted edges. Due to strong party polarization, the graph contains two large-scale communities that indicate, respectively, the Republican and Democratic parties. (Independents are incorporated into the community associated with the Democrats, in agreement with how those two Senators currently caucus.)

\begin{figure}[t]
\centering
  \includegraphics[width=.85\linewidth]{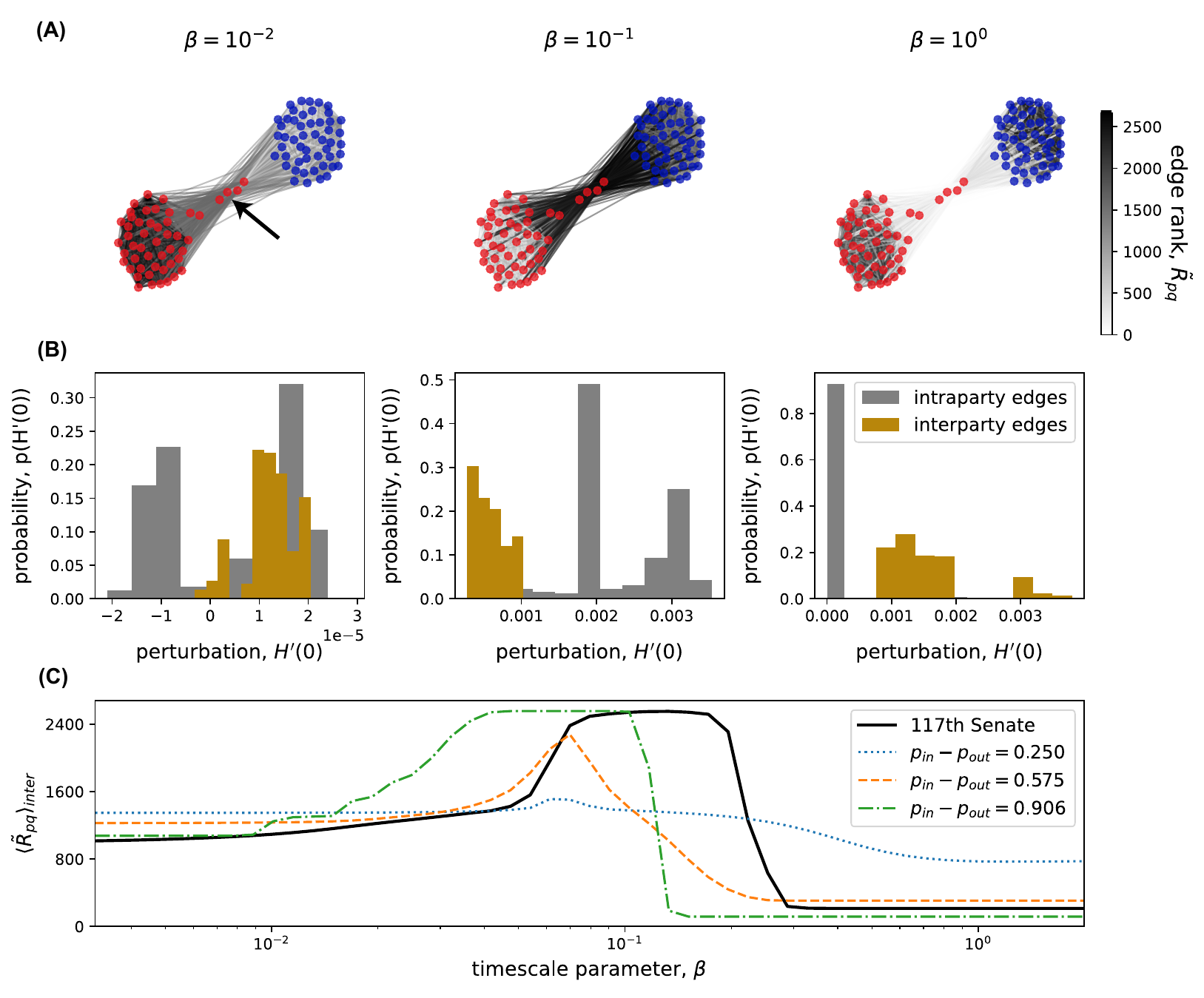}
  \vspace{-.5cm}
    \caption{ {\bf Rankings of inter- and intraparty edges for a polarized political network.}
	\textbf{(A)}~We  visualize the  rankings $\tilde{R}_{pq}$ of edges given by Algorithm~\ref{alg:approximate} using  grayscale   for a graph that encodes voting similarity among   the  117th U.S. Senate. The three columns depict three choices of timescale parameter  $\beta$. Nodes represent  Senators and the node colors red and blue indicate, respectively, the two major political parties: Republicans and Democrats (including independents).
   %
    \textbf{(B)}~For   the same three values of $\beta$, we plot empirically measured distributions $p(H'(0))$ of the approximate   change $H'(0)$  to diffusion-kernel VNE     given by Eq.~\eqref{eq:perturbation}. In each panel, we show two distributions: one in which $p(H'(0))$ is measured across intraparty edges, and one in which $p(H'(0))$ is measured across interparty edges.
    %
    \textbf{(C)} The solid black curve depicts the mean rank $\langle \tilde{R}_{pq}\rangle_{inter}$ across interlayer edges versus $\beta$. 
The colored curves indicate the values of $\langle \tilde{R}_{pq}\rangle_{inter}$ for a comparable random-graph model with two communities and different amounts of connectivity between communities (see text). 
Our main finding is that the interparty edges have the top rankings for larger $\beta$, the lowest rankings for intermediate $\beta$, and intermediate rankings for smaller $\beta$.
%
    }
\label{fig:congress_fig}
\end{figure}

In Fig.~\ref{fig:congress_fig}, we study the effect of $\beta$ on the rankings $\tilde{R}_{pq}$ given by Algorithm~\ref{alg:approximate} for intralayer and interlayer edges for the voting-similarity network.
In Fig.~\ref{fig:congress_fig}(A), we visualize the voting-similarity network and use edge color (in grayscale) to depict the $\tilde{R}_{pq}$ values for three choices for   $\beta\in\{10^{-2},10^{-1},10^0\}$. Senators' party affiliations (i.e., Republican vs.\  Democrats and independents) are indicated by the node colors (red vs.\ blue). Observe that the resulting network contains two well-separated communities due to party polarization.
That said, there is an arrow pointing to a set of Senators $\{$Susan Collins (R), Lisa Murkowski (R), Rob Portman (R), Mitt Romney (R), Shelley Capito (R),  Lindsey Graham (R)$\}$ whose voting patterns are not as strongly polarized along party lines. These Senators are found to have many edges to both Republicans and Democrats, and as such, their node degrees are approximately 12x higher than those for other Senators. 
%

Observe in the right-most column of Fig.~\ref{fig:congress_fig}(A) for $\beta=1$   that the top-ranked edges are the  {interparty edges} that connect Senators in different political parties. Interestingly, for $\beta = 10^{-1}$ (center column) the interparty edges are the ones with  the lowest rankings, and for $\beta=10^{-2}$ (left column) the rankings of interparty and {intraparty  edges} are similar. 
This nonlinear effect of $\beta$ on $\tilde{R}_{pq}$ is further supported in panels (B) and (C).
%
%
In Fig.~\ref{fig:congress_fig}(B), we display a probability distribution $P(H'(0))$ of our approximate perturbations $H'(0)$ given by Eq.~\eqref{eq:perturbation}  for the same three values of $\beta$.  Particularly, we separate the distributions according to interparty edges (gold) and intraparty edges (grey).  For small $\beta$, these two distributions have a similar support.  For medium $\beta$, we observe that   the  perturbations to VNE due to removing intraparty edges are larger than those for interparty edges.  Similar to panel (A), we observe for large $\beta$ that the ordering flips---that is,  the largest perturbations are now due to the removal of interparty edges as opposed to   intraparty edges.  These   findings   corroborate our network visualization in panel  (A).
%
Lastly, in Fig.~\ref{fig:congress_fig}(C) we use a solid black curve to plot  the average  ranking $\langle \tilde{R}_{pq} \rangle_{inter} $ of interparty edges as a function of $\beta$.  
As expected based on our findings in panels (A) and (B), $\langle \tilde{R}_{pq} \rangle_{inter} $ takes on an intermediate value for small $\beta$, is very large for a medium value of $\beta$, and then is very small for large $\beta$.

To gain {deeper} insight, in Fig.~\ref{fig:congress_fig}(C) we also plot    $\langle \tilde{R}_{pq} \rangle_{inter} $ values for a comparable random-graph model yield graphs with two communities, each containing 50 nodes (see the colored dotted, dot-dashed, and dashed  curves).
%
%
Specifically, we study  a  stochastic block model  with $N=100$ nodes and with two communities of equal size. Nodes $\{1,\dots,50\}$ are assigned to community 1 and nodes $\{51,\dots,100\}$ are assigned to community 2. Then, we create \emph{intracommunity edges}   between pairs of nodes in the same community uniformly at random with probability $p_{in}$ and \emph{intercommunity edges} are created between pairs of nodes in different communities uniformly at random with probability $p_{out}$. 
We fix the mean   $\left(p_{in} + p_{out}\right)/2$ so that the expected number of edges, $N(N-1)\left(p_{in} + p_{out}\right)/2$, matches that for the voting similarity network. We then considered three choices for  the difference $(p_{in} - p_{out})\in \{0.25,0.575,0.911\}$, which is a measure for structural polarization.
We note that the last value, 0.911,  closely matches the empirically measured difference in edge probability for intraparty and interparty edges for the voting-similarity network. And in fact,  stochastic block models are a popular graphical model for political polarization \cite{shekatkar2018ImportanceOI} and other sources for community structure \cite{peixoto2017nonparametric,abbe2017community}.
%
 In Fig.~\ref{fig:congress_fig}(C), we plot  $\langle \tilde{R}_{pq} \rangle_{inter} $ (now corresponding to intercommunity edges) for three choices of $(p_{in} - p_{out})$. 
Observe for the larger two values of $ (p_{in} - p_{out})$  that as $\beta$ increases, the average  $\langle \tilde{R}_{pq} \rangle_{inter} $ exhibits a peak and then becomes very small. This phenomenon is very similar to what we observe for the voting similarity network (black solid curve). Thus, this behavior for rankings $ \tilde{R}_{pq} $ appears to occur due to the presence of community structure (that is, and not due to some other possible structural property of the empirical voting-similarity network).


\subsection{Multimodal transportation system with  roads and metro lines that have different speeds}\label{sec:transit_network}
%
In our next case study, we study the   importance of edges for a multimodal transportation network that encodes roads and metro lines in London ~\cite{taylor2015contagion}. In this data (which is available at \cite{dane_code}), {2217} vertices represent intersections and there are two types of edges {($15$ metropolitan lines and $2854$ roads)} to encode these two {modes} of transportation. Notably, the metro lines connect 11 major stations in London, {and} their locations were mapped to the nearest road intersection.   

Given these two types of edges, in this section we study which has higher rankings, roads or metro lines, and we investigate how these rankings $\tilde{R}_{pq}$ change as one varies the relative speed between roads and metro lines. That is, in addition to studying how the rankings $\tilde{R}_{pq}$ change for different values of the timescale parameter $\beta$, we will also introduce    \emph{balancing parameter} $\chi\in[0,1]$ that controls the extent to which metro lines are faster or slower than roads. 
More specifically, we define an undirected, weighted transportation network with an adjacency matrix
\begin{align}
{\bf A} = \chi {\bf A}_{metro} + \left(1 - \chi\right) {\bf A}_{road},
\end{align}
where  ${\bf A}_{metro}$ and ${\bf A}_{road}$ are adjacency matrices in which edges only encode metro lines and roads, respectively. Note that there is no diffusion across metro lines (or roads) in the limit $\chi\to0$ (or $\chi\to1$).
Finally, we note that it is important to study different relative speeds of roads and metro lines, since it has been shown, for example,  to significantly influence the locations where congestions can occur \cite{strano2015multiplex}.
%

In Fig.~\ref{fig:transport_fig}(A), we provide a visualization of the multimodal transportation network with node locations reflecting their geospatial coordinates (i.e., latitudes and longitudes).  
Thick and thin lines represent metro lines and  roads, respectively, and the edge colors reflect  their rankings $\tilde{R}_{pq}$ for the choices $\chi = 0.6$  and $\beta=10^{-1}$.
In Fig.~\ref{fig:transport_fig}(B), we present a similar network visualization, except that we now consider a larger timescale parameter:  $\beta=10^0$. By comparing panels (A) and (B), observe for $\beta=10^{-1}$ that the top-ranked edges are metro lines. In contrast, metro lines  have the lowest rankings for $\beta=10^{0}$.


In Fig.~\ref{fig:transport_fig}(C), we plot
the average ranking $\langle \tilde{R}_{pq}\rangle_{metro}$ across metro lines versus the timescale parameter $\beta$. The three thick colored curves reflect three choices for $\chi$. The thin curves depict the $\tilde{R}_{pq}$ values for each individual metro line. 
In Fig.~\ref{fig:transport_fig}(D), we show  similar information by plotting $\langle \tilde{R}_{pq}\rangle_{metro}$ versus $\chi$ for a few choices of $\beta \in [0.01,3]$.
Together, panels (C) and (D) highlight that the metro lines have the top rankings when $\beta$ is sufficiently small and $\chi $ is sufficiently large.  That is, transport across metro lines must be sufficiently faster than that across roads, and one only   considers diffusion at a short timescale.

From a topological perspective, metro lines introduce long-range, `short-cut' connections across a road network that is largely a two-dimensional geometric substrate \cite{taylor2015contagion}. We find it   interesting that the   importance of such long-range edges {is} so sensitive to timescale. We hypothesize that this phenomenon may be related to the following observation: if someone moves between two spatially distant nodes  across a small timescale $\beta$, then they  must have traversed a metro line. This  is no longer true   at a large timescale $\beta$, since there are many paths that one  could take. Thus, one might expect short-cut edges to have a greater impact on diffusion dynamics (and rankings derived therefrom) at shorter timescales versus longer timescales.

\begin{figure}[t]

\centering
  \includegraphics[width=1\linewidth]{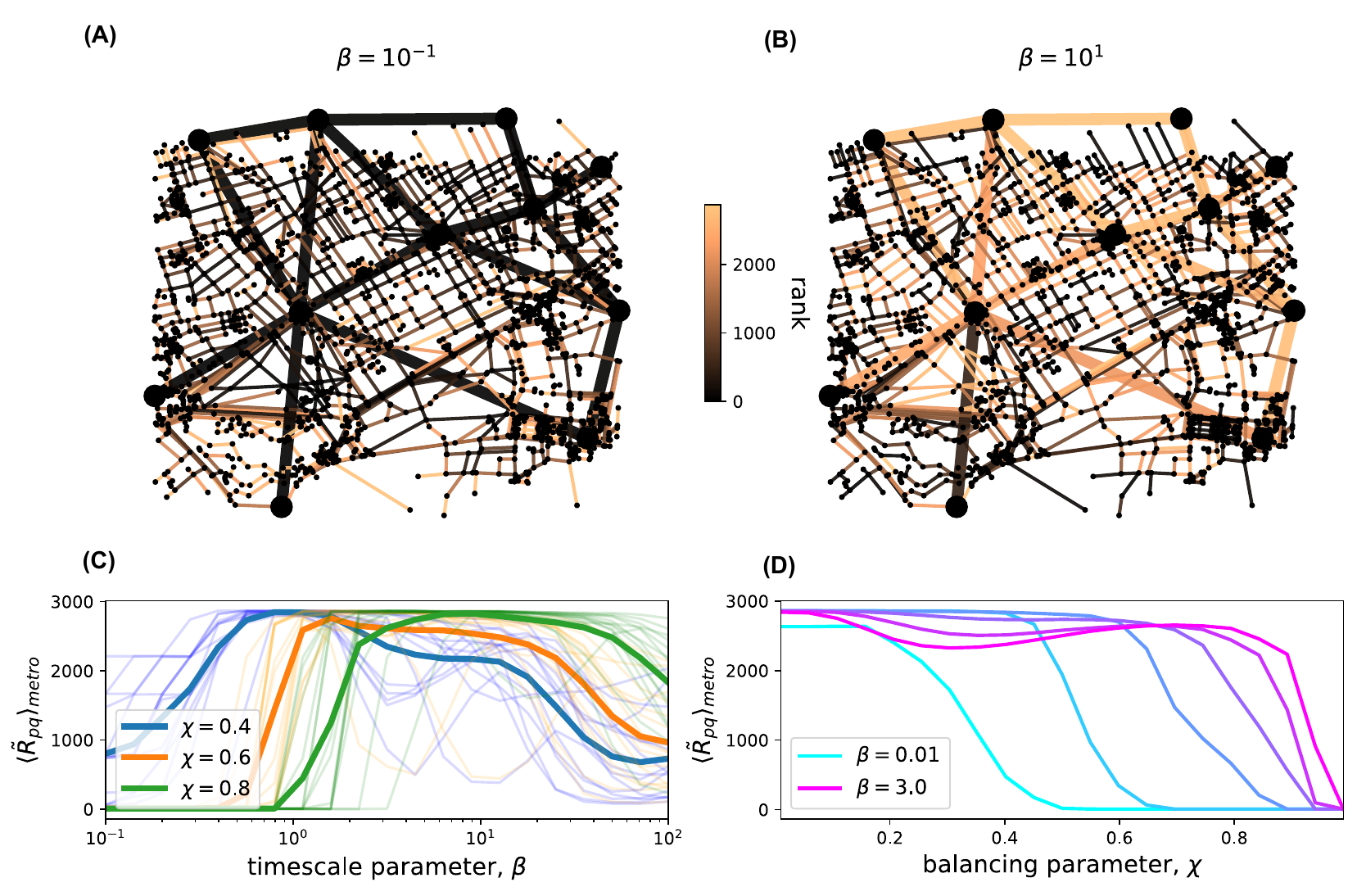}
  \vspace{-.7cm}
    \caption{{\bf Transportation network with  metro lines and roads having different speeds.}  \textbf{(A)} Visualization of a London transportation system in which metro lines and roads are depicted by thick and thin lines, respectively. Edge colors depict their rankings $\tilde{R}_{pq}$ for a small timescale $\beta=10^{-1}$.
\textbf{(B)} Same information as panel (A), except for a larger timescale $\beta=10^{0}$. Observe that the metro lines have the highest  rankings in (A), but the lowest rankings in (B).
%
   \textbf{(C)} Thick curves depict the average edge ranking $\langle \tilde{R}_{pq}\rangle_{metro}$ across metro lines versus $\beta$ for three choices of the \emph{tuning parameter} $\chi$ (which controls whether diffusion is faster over metro lines or roads). Each  thin curve indicates how the ranking $\tilde{R}_{pq}$ for each metro line varies with $\beta$.
   \textbf{(C)} Similar information as panel (C), except we now plot   $\langle \tilde{R}_{pq}\rangle_{metro}$ versus $\chi$, and different curves indicate  several  choices of $\beta$.
   In general, we find that metro lines (which introduce long-range ``short-cuts' across the road network) are the top-ranked edges provided that $\chi$ is sufficiently large and $\beta$ is sufficiently small.
    }
\label{fig:transport_fig}
\end{figure}

\subsection{Multiplex Brain Network with Interlayer Coupling}\label{sec:brain_network}
In our final case study, we study the   importance of edges for a multiplex brain network in which  nodes represent brain regions and different \emph{network layers} encode coordinated brain activity patterns at different frequency bands. Our focus of this section is to investigate how the edge rankings $\tilde{R}_{pq}$  change as one varies the strength $\omega\ge0$ of coupling between layers (as well as the timescale parameter $\beta$).

We study an empirical network   obtained  from \cite{brain_data} that was constructed from functional magnetic resonance imaging (fMRI) data   in which nodes encode   regions of interest (ROI) in human brains and edges encode a measure for spectral coherence between fMRI signals that are aggregated across each ROI
\cite{guillon2017loss}.   
In particular,   the frequency content of an fMRI signal for each ROI was decomposed into different frequency bands, and then coherence relationships were identified for each frequency band. The result was a  set of \emph{network layers} in which each layer encodes spectral-coherence relationships at a particular frequency band.

\begin{figure}[t]
\centering
  \includegraphics[width=1\linewidth]{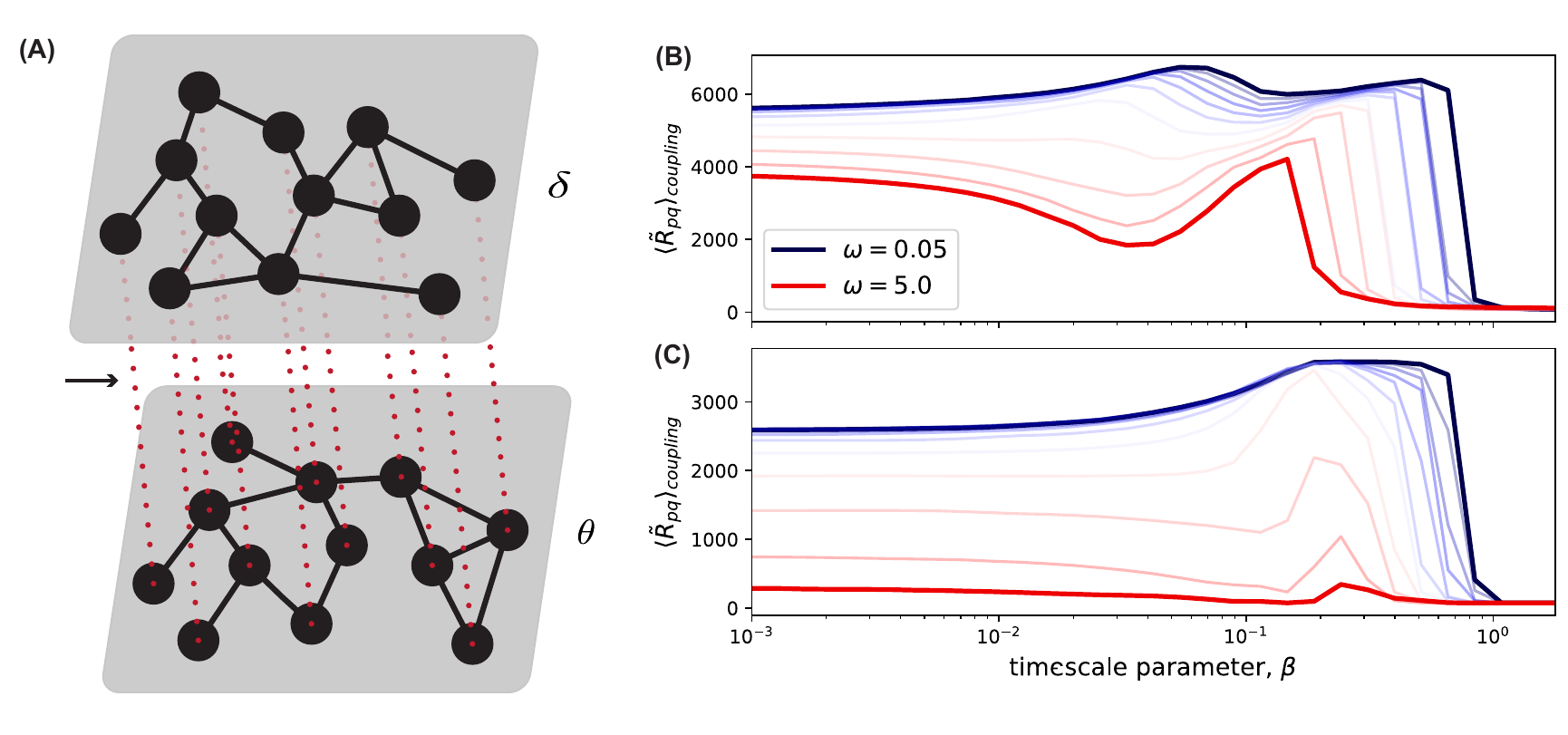}
  \vspace{-.5cm}
    \caption{{\bf Multiplex network representing coherent fMRI signals among brain regions.}
    \textbf{(A)} Toy illustration of a multiplex network with two layers representing two frequency bands: $\delta$ and $\theta$ waves. The solid black  and red dotted lines represent, respectively, intralayer and interlayer edges. Interlayer edges are weighted by a coupling strength $\omega\ge0$.
    \textbf{(B)} Average ranking $\langle \tilde{R}_{pq} \rangle_{inter}$  across interlayer edges versus  $\beta$ for several choices of $\omega\in[0.05,5]$. 
    \textbf{(C)} Similar information as panel (B), except we now consider a synthetic model in which each layer is an Erd\H{o}s-R\'enyi $G_{NM}$ random graph in which $N$ and $M$ match the values for the empirical  network layers.
    Observe   a similar trend in panels (A) and (B):  interlayer edges have the top rankings when $\beta$ is very large, and these rankings take on their lowest values when $\beta$ has some intermediate value.}
\label{fig:brain_fig}
\end{figure}

Focusing on the frequency bands of $\delta$-waves (2-4 Hz) and $\theta$-waves (4.5-7.5 Hz), we considered
two adjacency matrices  ${\bf A}^{(\delta)}$ and ${\bf A}^{(\theta)}$, respectively.
Both matrices are size $N=148$ and encode two types of relationships among a common set of nodes $\mathcal{V}=\{1,\dots,N\}$. 
That is, a nonzero matrix entry ${A}^{(\delta)}_{ij}$ encodes a relation between brain regions $i$ and $j$ at the $\delta$ frequency band, whereas  ${A}^{(\theta)}_{ij}$ encodes a similar spectral-coherence relationship, but at the $\theta$ frequency band. To focus our study on sparse networks, we also applied a threshold to these matrices as described in Appendix~\ref{app:brain_weights}. This resulted in each layer containing $M=1,652$ undirected, weighted edges.

Given these two network layers encoded by   ${\bf A}^{(\delta)} $ and ${\bf A}^{(\theta)} $, which we refer to as \emph{intralayer adjacency matrices}, we   then constructed a two-layer multiplex network by coupling them together. Specifically, we consider a \emph{supra-adjacency matrix} 
\begin{equation}\label{eq:supra}
     \mathbb{A} = \begin{bmatrix} {\bf A}^{(\delta)} & 0 \\ 0 & {\bf A}^{(\theta)} \end{bmatrix} + \omega \begin{bmatrix} 0 & {\bf I} \\ {\bf I} & 0 \end{bmatrix},
\end{equation}
where ${\bf I}$ is the identity matrix and $\omega\ge 0$ is an \emph{interlayer coupling strength} that tunes how strongly the layers are coupled together. Coupling matrices in this way is called uniform, diagonal coupling.

As visualized in  Fig.~\ref{fig:brain_fig}(A), coupling network layers following Eq.~\eqref{eq:supra}  introduces \emph{interlayer edges} (red dotted lines) that connect each node $i$ in the first layer to itself in the second layer. (The two network layers are indicated by the gray shaded regions, and the set of nodes is the same for both layers.) That is,   there is one interlayer edge connecting each node in one layer to itself in the other layer. 
In contrast, we will   refer to the edges encoded by  ${\bf A}^{(\delta)} $ and ${\bf A}^{(\theta)} $ as \emph{intralayer edges}. Notably, interlayer edges have a weight given by $\omega$, whereas the possible weights for intralayer edges are encoded by the matrix entries in ${\bf A}^{(\delta)} $ and ${\bf A}^{(\theta)} $.

Similar to our study in Sec.~\ref{sec:congress_network} (where we examined the rankings of interparty edges), we will now pay particular attention to the rankings $\tilde{R}_{pq}$ of interlayer edges. We ask the following: When do the interlayer edges have the top rankings?
In Fig.~\ref{fig:brain_fig}(B), we plot the average ranking $\langle \tilde{R}_{pq} \rangle_{inter}$  across interlayer edges versus the timescale parameter $\beta$. Different curves indicate different choices of the interlayer coupling strength $\omega\in[0.05,5]$. Observe that the interlayer edges have the top rankings only when $\beta$ is very large. Moreover, for intermediate values of $\beta$,  there exists a peak that corresponds to when the interlayer edges obtain their lowest rankings. Interestingly, these two phenomena are reminiscent of our findings for interparty edges that we presented in Sec.~\ref{sec:congress_network}. By comparing the different curves for different  $\omega$, we can observe that the  peak is sharper  for larger $\omega$ values, whereas it is more of a  plateau for smaller $\omega$ (and can even have two optima). In Appendix~\ref{app:brain_weights}, we show that qualitatively similar results occur for 24 other multiplex brain networks taken from \cite{guillon2017loss}.


In Fig.~\ref{fig:brain_fig}(C), we support these findings by studying $\langle \tilde{R}_{pq} \rangle_{inter}$   for a random-graph model in which the two layers are given by  Erd\H{o}s-R\'enyi  $G_{NM}$ networks for $N = 148$ and $M = 7,052$ in which the number of edges $M$ is chosen so that the numbers of intralayer edges in each layer matches those for ${\bf A}^{(\delta)} $ and ${\bf A}^{(\theta)}$. 
We plot $\langle \tilde{R}_{pq} \rangle_{inter}$ versus $\beta$ for this generative model in Fig.~\ref{fig:brain_fig}(C) for different choices of $\omega$, where one can observe that the behavior is qualitatively similar that which was shown in Fig.~\ref{fig:brain_fig}(B). Despite this similarity, there are some differences between the curves panels (B) and (C), which could result because the empirical network contains rich structural features (e.g.,  heterogeneous  weights for intralayer edges and non-random edges) that are lacking from the simple model that we study in panel (C).

\subsection{{Comparison to other edge centralities for the case studies}}\label{app:centrality_comparison}
{Here, we provide additional insight about our VNE-based edge centrality  by comparing Algorithm~\ref{alg:approximate} to three other edge centralities:
current flow \cite{brandes2005centrality}, betweenness \cite{newman2010networks}, and an extension of  node degree centrality whereby a centrality score is computed for each edge $(i,j)$ by adding  the degrees of nodes $i$ and $j$. 
We make this comparison for   three networks from the three case studies studied above.


%

\begin{figure}[h]
\label{fig:centrality_comparison}
\centering
  \includegraphics[width=1\linewidth]{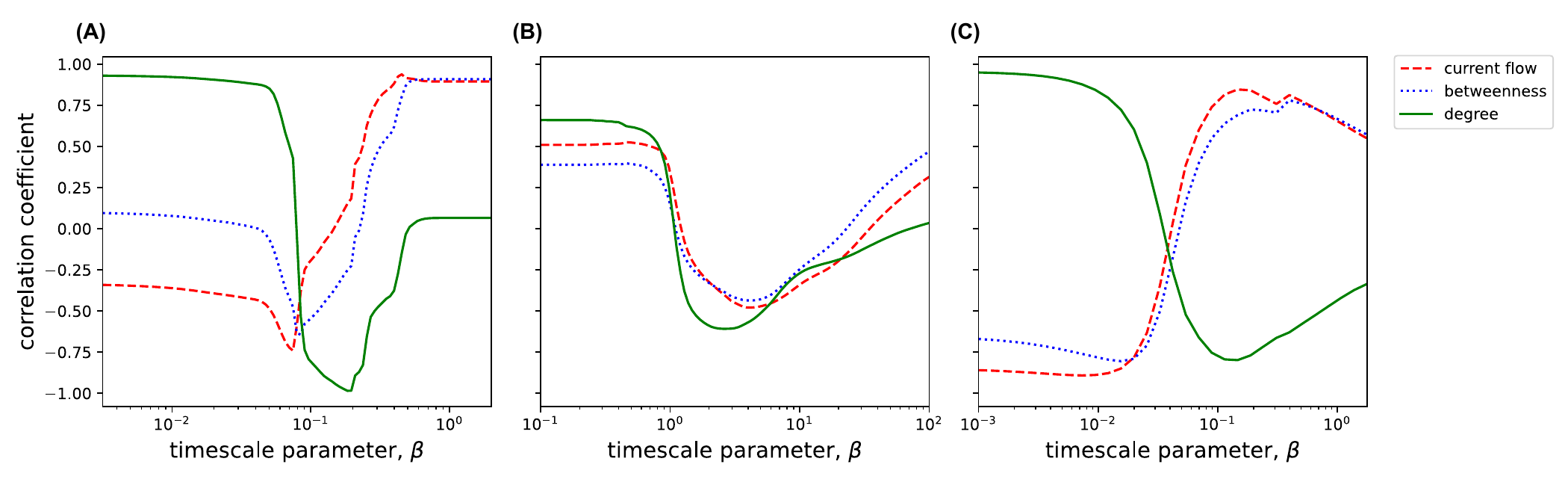}
  \vspace{-.2cm}
    \caption{{{\bf Comparison of Algorithm~\ref{alg:approximate} to other edge centralities.}
%
    \textbf{(A)}~For the U.S. Senate voting-similarity network from Sec.~\ref{sec:congress_network}, we plot the Pearson correlation coefficient between rankings according to Algorithm~\ref{alg:approximate} with varying $\beta$  and three other centralities:  edge current flow (red), betweenness (blue), and a degree-based centrality (green).     
    \textbf{(B)}~Similar results but for the multimodal transportation network from Sec.~\ref{sec:transit_network}  under the choice of $\chi = 0.6$. 
    \textbf{(C)}~Similar results but for the multiplex brain network from Sec.~\ref{sec:brain_network} under the choice of $\omega = 1$.     
    }}
\end{figure}

In Fig.~\ref{fig:centrality_comparison}(A), we plot Pearson correlation coefficients relating rankings according to Algorithm~\ref{alg:approximate}   to rankings according to each of the three other centralities. Results are shown for a U.S. Senate voting-similarity network (Sec.~\ref{sec:congress_network}).
Observe for the smaller values of $\beta$ that  rankings  according to Algorithm~\ref{alg:approximate} are strongly positively correlated with rankings under   edge degree centrality  but  are uncorrelated (or weakly correlated) with rankings according to  either edge current flow and betweenness. In contrast, for the larger values of $\beta$, rankings according to Algorithm~\ref{alg:approximate} become    uncorrelated with rankings according to edge degree centrality, but they are strongly correlated with rankings according to edge betweenness and current flow.

In Fig.~\ref{fig:centrality_comparison}(B), we plot the same information as in panel (A) but we now consider   a multimodal transport network (Sec.~\ref{sec:transit_network}).
Observe that rankings according to Algorithm~\ref{alg:approximate} are weakly  correlated for most choices of  $\beta$, and these can be either positively or negatively correlated depending on   $\beta$. Interestingly, the correlation between  Algorithm~\ref{alg:approximate} and the other edge centralities are very similar at each value of $\beta$.

In Fig.~\ref{fig:centrality_comparison}(C), we again plot the same information but we now consider a multiplex brain network (Sec.~\ref{sec:brain_network}).
Observe for the smaller values of  $\beta$ that rankings according to Algorithm~\ref{alg:approximate} are strongly correlated with edge degree, but they are strongly negatively correlated with rankings according to both edge current flow and betweenness.  Further, for the larger values of $\beta$, rankings according to Algorithm~\ref{alg:approximate} are weakly negatively correlated with rankings according to edge degree, but they are strongly positively correlated with rankings according to  edge current flow and betweeness.
%

From these observations, we conclude that common notions of centrality can be similar to our proposed centrality measure that quantifies the spectral complexity of diffusion, but only for some networks and only when the diffusion timescale lies within some range of values. Importantly, our proposed VNE-based edge rankings of edges are notably different from these known rankings 
 for other values of $\beta$, particularly, mid-range values of $\beta$.
We find that the intricacies of such rankings and their correlations are generally difficult to relate to network properties, and we leave this pursuit open to future work.}

\section{Discussion}\label{sec:Discuss}
%
In this paper, we have proposed a centrality measure derived from von Neumann entropy (VNE) to identify and rank  network edges according to their impact on the spectral complexity of information diffusion over a graph.
Our work   complements prior work \cite{ghavasieh2021unraveling} that 
{uses VNE to obtain a node centrality measure and which   does not use first-order approximation theory to be scalable for large graphs}.
In Sec.~\ref{sec:Rank}, we developed a measure for the importance of each edge $(p,q)\in\mathcal{E}$ by considering how VNE would change upon its removal. We rank the edges according to these VNE perturbations, and the resulting top-ranked edge is the one such that its removal would most increase VNE (and subsequently, most increase the spectral complexity of information diffusion at timescale $\beta$).
Our approach complements existing centrality measures  that are related to information spreading  
\cite{freeman1977set,katz1953new,estrada2009communicability,estrada2010network,kempe2003maximizing,kitsak2010identification,lehmann2018complex} but which do not consider the spectral complexity of information diffusion, as measured by VNE.
Because the  direct computation of   VNE perturbations are computationally expensive (or infeasible) for large-scale graphs, in Sec.~\ref{sec:Rank2} we developed an approximate ranking of edges based on first-order approximation theory that can efficiently predict  these VNE perturbations. That said, we note that the computational challenges are not fully resolved for VNE-based analyses of graphs. First, VNE requires one to compute a full set of matrix eigenvalues, which has a general scaling of $\mathcal{O}(N^3)$ for graphs with $N$ nodes. In principle, this poor computational scaling could be mitigated by incorporating spectral approximation techniques such as those that rely on message passing \cite{cantwell2019loops},
subgraph motifs \cite{preciado2010local,preciado2012moment}, kernel polynomial methods \cite{weisse2006kernel,dong2019network}, random matrix theory 
\cite{peixoto2013eigenvalue}, {and randomized numerical linear algebra \cite{martinsson2020randomized}.}


Also, for our goal of approximating how VNE is perturbed due to structural modifications to a network, it may be practical to approximate VNE perturbations using only a subset of eigenvalue perturbations, alleviating the need to compute all eigenvalues. That said, the number of  eigenvalues required to maintain some level of approximation accuracy would vary greatly depending on the diffusion timescale $\beta$ that is considered. For example, when $\beta$ is very large, one would expect that $H'(0)$ given by Eq.~\eqref{eq:perturbation} could be accurately approximated using only a few of the smallest eigenvalues $\lambda_i$ (i.e., since $e^{-\beta \lambda_i}$ would be negligible for larger eigenvalues). Here, we intentionally explored a wide range of timescales   $\beta$, which can also be computationally expensive. Thus, it may also be beneficial to develop methods that explore the $\beta$-parameter space more systematically by incorporating, e.g., one's external knowledge about some dynamical, structural, or preferential criteria. 
{One possible direction would be to  leverage spectral theory for  block-structure matrices including  those arising  for multiplex  networks \cite{gomez2013diffusion,sole2013spectral,taylor2017eigenvector,taylor2019tunable,taylor2020multiplex}.
}

Despite these challenges, we utilized our approximation theory for perturbed diffusion-kernel VNE to study the entropic importance of edges via three case studies with empirical networks to investigate structural/dynamical mechanisms that can systematically influence the entropic importance of edges. In Sec.~\ref{sec:congress_network}, we studied a network encoding  voting similarity  in the 117th U.S. Senate and studied the importance of interparty edges (i.e., those that connect across the two large-scale communities that result from party polarization). In Sec.~\ref{sec:transit_network}, we studied  a multimodal transportation network and studied the importance of edges representing   metro lines as opposed to those representing roads. In Sec.~\ref{sec:brain_network}, we studied a   multiplex brain network and studied the importance of interlayer edges that couple network layers versus intralayer edges that connect nodes within a particular layer. In each of these case studies, we found that   the edges that are deemed to be the most important drastically change by considering different diffusion timescales (i.e., different values of the   parameter $\beta$).

Notably, our proposed techniques for measuring the importance of edges with respect to VNE is expected to have many applications beyond our three cases studies. As an entropy-based measure, we expect our techniques to help reveal new ways to leverage information theory to study diverse types of biological, social, and physical  systems. Because diffusion-kernel VNE has a well-defined interpretation for the diffusion equation [see Eq.~\eqref{eq:heat}], it can also contribute to the related literature on Laplacian-based algorithms for data analysis and machine learning  \cite{belkin2001laplacian,coifman2006diffusion,kloster2014heat,jeub2015think,bronstein2017geometric}. As one example, our ranking of edges could contribute towards the development of Laplacian-based  algorithms for graph sparsification \cite{spielman2011graph} as well as link prediction \cite{kunegis2009learning,torres2020glee}. {However, there are limiting cases for our proposed technique such as signed networks and directed networks, where the spectra can take on complex values for which VNE is not yet defined.  In addition, since our choice of VNE is defined by the diffusion process, it may be inappropriate to apply to networks with disconnected components, since diffusion can be trapped onto a component depending on the initial condition.}

{That said,} our investigation of the multiscale aspects of entropic importance in this paper (i.e., whereby the most important edges depend sensitively on the diffusion timescale $\beta$ that is considered) does also open up new   lines of scientific inquiry. For instance, we observed `ranking regimes' in which the entropic importances of edges remain relatively insensitive to $\beta$ as well as `regime transitions' in which rankings drastically change as $\beta$ varies. (See \cite{taylor2017eigenvector,taylor2019tunable} for further discussion on mutliscale centrality regimes and transitions.) Our case studies explored three structural scenarios in which such behavior can arise, and it would be beneficial to develop a deeper understanding in future research. Moreover, while we focused herein on a formulation of VNE that is related to information diffusion, our methods could be extended to other applications by considering other dynamical systems and their related matrices. For example, by considering a normalized Laplacian matrix $\tilde{{\bf L}} = {\bf D}^{-1}{\bf L}$, VNE-based methods  would  more appropriately describe information theory related to continuous-time random walks \cite{de2016spectral} as opposed to  diffusion. One could also consider spectral entropies relating to adjacency matrices and non-backtracking matrices \cite{alon2007non}.  {However, one benefit from defining VNE according to the diffusion equation is that it has a closed-form, analytical solution using spectral theory, and so spectral perturbation theory can be used as it is in this paper.  Extending to different dynamical processes may pose a challenge  if no closed-form analytical solution exists.  In such situations, methods such as trace estimators \cite{ubaru2017estimator} might help.}
Another direction could involve using Hodge Laplacian matrices \cite{schaub2020random,ziegler2022hodge}
to extend VNE and VNE-related centrality measures to simplicial complexes (which are a higher-order generalization of graphs). Common to all of  these scenarios, it would be beneficial to consider how structural modifications impact VNE, how these VNE perturbations can be used for centrality analysis, and how spectral perturbation theory can support the development of computationally efficient methodologies.

~

~

\appendix

\section{Proof of Theorem~\ref{def:gen_FOP}} \label{app:proof_vne}
%
Equation~\eqref{eq:partials} establishes that the first-order derivative of VNE can be expanded in terms of partial derivatives
$$
H'(0) = \left. \sum_{i,j} \dfrac{\partial h}{\partial f_i}\dfrac{\partial f_i}{\partial \lambda_j}  \dfrac{\partial \lambda_j}{\partial \epsilon} \right|_{\epsilon=0},
$$    
where we define VNE as $h({\bf f}) = -\sum_i f_i(\vec{\lambda}) \log_2(f_i(\vec{\lambda}))  $,
Our proof only requires that we solve for these terms. It was shown in \cite{li2018network} that the first terms are given by
%
\begin{align}
\dfrac{\partial h}{\partial f_i} &=  -\log_2\left(f_i( \vec{\lambda})\right)  - \frac{1}{\ln\left(2\right)}.
\end{align}
Considering the diffusion-kernel VNE given by 
$$
f_i(\vec{\lambda}) = \frac{e^{-\beta\lambda_i }}{\sum_k e^{-\beta\lambda_k }},
$$
we  obtain the derivatives $\partial f_i / \partial \lambda_j$ separately for $j=i$ and $j\not=i$. 
For $j\not=i$, we find
\begin{align}\label{eq:i_j}
\dfrac{\partial f_i}{\partial \lambda_j}  
&=
\frac{ \left( \sum_k e^{-\beta\lambda_k}\right) (0) - (e^{-\beta\lambda_i}) \left(-\beta e^{-\beta\lambda_j}  \right)}
{\left[\sum_k e^{-\beta\lambda_k}\right]^2}\nonumber\\
&=
0 + \beta  \dfrac{e^{-\beta\lambda_i}}{ \sum_k e^{-\beta\lambda_k} } \dfrac{e^{-\beta\lambda_j}}{ \sum_k e^{-\beta\lambda_k} } \nonumber\\
&= 
 \beta   f_i( \vec{\lambda}) f_j( \vec{\lambda}) .
\end{align}
For $j=i$, we find
\begin{align}\label{eq:i_i}
\dfrac{\partial f_i}{\partial \lambda_i}  
&=
\frac{ \left( \sum_k e^{-\beta\lambda_k}\right) (-\beta  e^{-\beta\lambda_i}) - (e^{-\beta\lambda_i}) \left(-\beta e^{-\beta\lambda_i}  \right)}
{\left[\sum_k e^{-\beta\lambda_k}\right]^2}\nonumber\\
&=
-\beta \left[\frac{ \left( \sum_k e^{-\beta\lambda_k}\right) (  e^{-\beta\lambda_i})}{\left[\sum_k e^{-\beta\lambda_k}\right]^2} - \frac{(e^{-\beta\lambda_i}) \left(  e^{-\beta\lambda_i}  \right)}
{\left[\sum_k e^{-\beta\lambda_k}\right]^2}\right]\nonumber\\
&=
- \beta  f_i( \vec{\lambda}) \left( 1- f_i( \vec{\lambda})   \right) .
\end{align}
Combing these results, we obtain
\begin{align}
\sum_j \dfrac{\partial f_i}{\partial \lambda_j}\dfrac{\partial \lambda_j}{\partial \epsilon}  
&= -\beta f_i( \vec{\lambda})  \left[\left(1 - f_i( \vec{\lambda}) \right) \frac{\partial \lambda_i}{\partial \epsilon} 
- \sum_{j \neq i}^{N} f_j( \vec{\lambda}) \frac{\partial \lambda_j}{\partial \epsilon}   \right]\nonumber\\
&= -\beta f_i( \vec{\lambda})  \left[ \frac{\partial \lambda_i}{\partial \epsilon}   
- \sum_{j =1}^{N} f_j( \vec{\lambda}) \frac{\partial \lambda_j}{\partial \epsilon}  \right].
\end{align}
Finally,  we  use that
\begin{equation}
\left. \dfrac{\partial \lambda_i}{\partial \epsilon} \right|_{\epsilon=0} = \lambda_i' (0) =({\bf u}^{(i)})^T \Delta  {L} {\bf u}^{(i)}
\end{equation}
is given by Eq.~\eqref{eq:eigen_perturb_b}.
We combine these terms to recover Eq.~\eqref{eq:perturbation}.


\section{Numerical validation of Theorem ~\ref{def:gen_FOP}} \label{app:valid}
In Fig.~\ref{fig:pert_valid}, we provide numerical validation for Thm.~\ref{def:gen_FOP}. First, we computed  the  VNE $h(L)$ of a random graph that was sampled from the Erd\H{o}s-R\'enyi $G_{NM}$  model with $N=100$ nodes and $M=2000$ edges.   In Fig.~\ref{fig:pert_valid}(A), we plot the perturbed VNE $h(L+\Delta L)$ (blue dotted curve) after removal of $k$ edges with $k\in\{0,\dots,20\}$.  Edges were selected for removal uniformly at random. For comparison, the red dashed curve depicts our first-order approximation $h(L+\Delta L) \approx   h(L) + H'(0)$ with $H'(0)$ given by Eq~\eqref{eq:perturbation} {with $\Delta L$ defined according to Proposition \ref{cor:1edge}}. Note  {for the deletion of several edges that 
$\Delta L=-\sum_{(p,q)}\Delta L^{(pq)}$, where the sum is taken over the set of removed edges  \cite{taylor2016synchronization}.}
{Observe in   Fig.~\ref{fig:pert_valid}(A)} that, as expected, the first-order approximation is more accurate when the perturbation is smaller (i.e., smaller $k$).
%

\begin{figure}[t]
\centering
\includegraphics[width=1\linewidth]{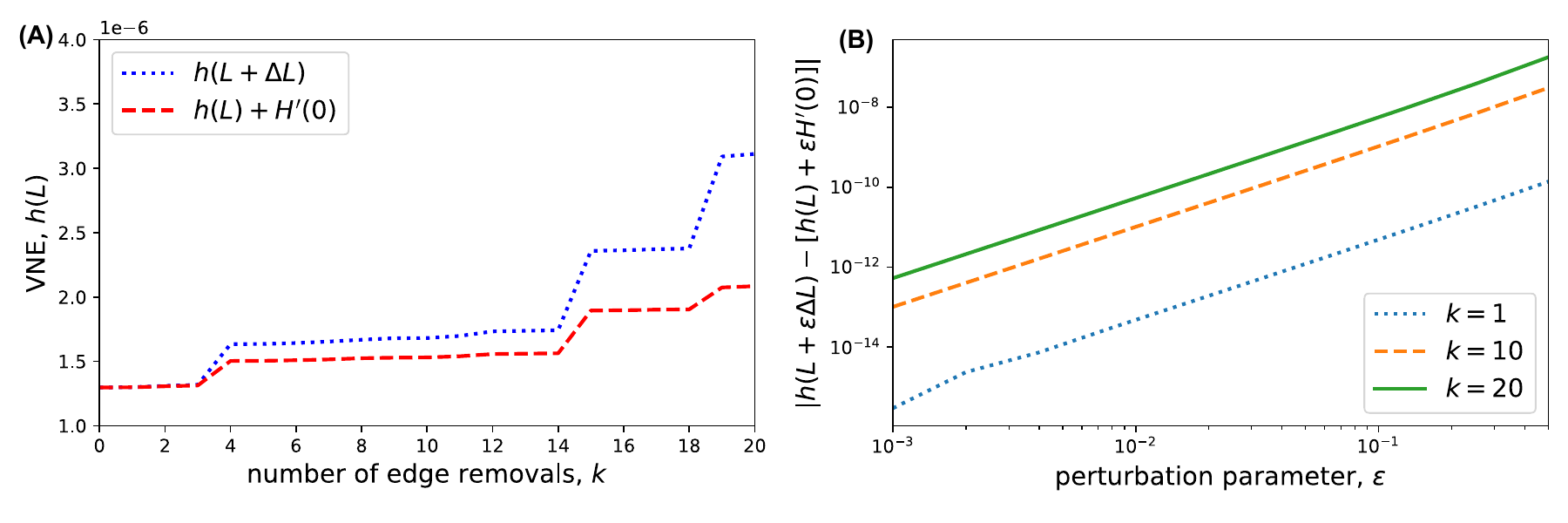}
\vspace{-.9cm}
\caption{ {\bf Numerical validation of first-order perturbation  theory for   diffusion-kernel VNE.}
{\bf (A)}~Comparison of   $h(L+  \Delta L)$ and our first-order prediction, $h(L) + H'(0)$, with $H'(0)$ given by Thm.~\ref{def:gen_FOP},  for a perturbed graph. We   consider an Erd\H{o}s-R\'enyi $G_{NM}$ random   graph with $N = 100$ nodes and $M = 1500$ unweighted edges, and each perturbation matrix $\Delta L$   encodes the removal of $k$ edges {as described in \cite{taylor2016synchronization} (which generalized Lemma \ref{cor:spec_pert} to edge sets).}
Observe that the first-order approximation becomes less accurate for larger perturbations (i.e., larger $k$).
{\bf (B)}~Numerical validation that the approximation error $E(\epsilon) \equiv   |h(L+\epsilon\Delta L) - [h(L) + \epsilon H'(0)]|$ vanishes in the limit of small $\epsilon$ as $\mathcal{O}(\epsilon^2)$, that is, $\log(E(\epsilon) ) \varpropto 2\log(\epsilon)$. 
We plot $E(\epsilon)$ for the same network as in panel (A) in a log-log scale and compute a least-squares linear fit to obtain empirically measured slopes
$\{2.042, 2.022, 2.036\}$, respectively, for three choices of $k\in\{ 1, 10, 20\}$. These are all nearly equal to the predicted slope of 2, supporting our claim that Thm.~\ref{def:gen_FOP} has second-error error.
}
\label{fig:pert_valid}
\end{figure}

In Fig.~\ref{fig:pert_valid}(B), we support our claim that the approximation $h(L+\epsilon \Delta L) \approx   h(L) + \epsilon H'(0)$ is first-order accurate by showing that the approximation error  has a second-order scaling behavior. Specifically, we plot the approximation error $E(\epsilon) \equiv |h(L+\epsilon \Delta L) - [   h(L) + \epsilon H'(0)]|$ versus $\epsilon$ and find  that the error decays as $\mathcal{O}(\epsilon^2)$. Note that in a log-log scale, this scaling corresponds to a linear relationship $\log( E(\epsilon) )\varpropto 2 \log(\epsilon)$. In Fig.~\ref{fig:pert_valid}(B), we plot $\log( E(\epsilon) )$ versus $ \log(\epsilon)$ for three choices of $k\in\{1,10,20\}$, and  we computed a least-squares linear fit to each curve. We empirically observed the slopes to be approximately 2.014, 2.022, and 2.026 for these three lines, respectively. Because these are all very close to our analytically predicted slope of 2, we can be confident that our first-order approximation does in fact have second-order error.

{
\section{Further Comparison of Algorithms~\ref{alg:actual} and \ref{alg:approximate}} \label{app:valid2}
Here, we extend the results shown in Fig.~\ref{fig:runtime}(B) by repeating the experiment for two empirical networks: a network encoding voting-pattern similarity among the  U.S. Senate and a multiplex network derived from fMRI data. See Secs.~\ref{sec:congress_network} and \ref{sec:brain_network}, respectively, for their descriptions. We also highlight that we were unable to  conduct this experiment for the transit network described in Sec.~\ref{sec:transit_network}, since we found it infeasible to compute Algorithm~\ref{alg:actual} for that network. Specifically, the transit network contains approximately 22 times more nodes than the other two empirical networks. Since the runtime of Algorithm~\ref{alg:actual} scales like $MN^3$ for $M$ edges and $N$ nodes, we estimate that  Algorithm~\ref{alg:actual} for the transit network could take up to 10,000 times longer than for the other two empirical networks.  (Recall that the poor scaling of Algorithm~\ref{alg:actual} was our main motivation for developing Algorithm~\ref{alg:approximate}.)

\begin{figure}[h]
    \label{fig:compare2}
    \centering
    \includegraphics[width=\linewidth]{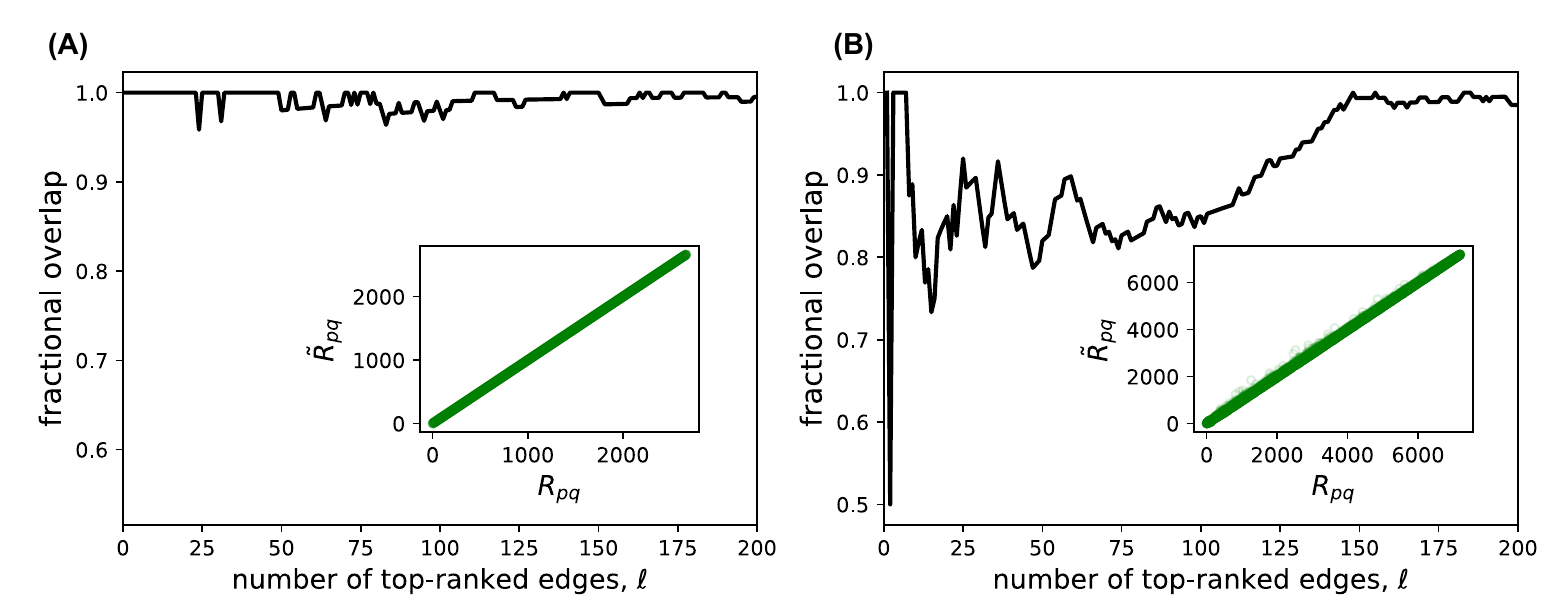}
    \vspace{-.3cm}
    \caption{{ {\bf Comparison of Algorithms~\ref{alg:actual} and \ref{alg:approximate} for   U.S. Senate    and multiplex brain networks.}
        Extending the results in Fig.~\ref{fig:runtime}(B), we now compare the two algorithms 
        for two empirical networks:  \textbf{(A)}~a U.S. Senate voting network and \textbf{(B)}~a multiplex brain network with interlayer coupling $\omega = 1$. See Secs.~\ref{sec:congress_network} and \ref{sec:brain_network}, respectively, for their descriptions. In both panels, we  plot the fractional overlap among the top-$\ell$ ranked edges for the two algorithms with $\beta=1$ and varying $\ell$. The insets show scatter plots that directly compare the  rankings from the algorithms.
    }}
\end{figure}

In Fig.~\ref{fig:compare2}(A) and (B), we present results that are similar to those shown in Fig.~\ref{fig:runtime}(B) but which are now obtained for  two empirical networks: (A) a U.S. Senate voting network and (B) a multiplex brain network. For both networks, we plot   the fractional overlap among the top-$\ell$ ranked edges for the two algorithms with $\beta=1$ and varying $\ell$. Observe for both networks that both algorithms identify a similar set of top-ranked edges. The insets show scatter plots that directly compare  the algorithms rankings. Observe that the approximate edge rankings $\{\tilde{R}_{pq}\}$ and true edge rankings $\{{R}_{pq}\}$ are very similar for all edges in both networks. These results recapitulate our findings that were previously  discussed for Fig.~\ref{fig:runtime}(B).
}

\section{Further description of the voting similarity network in Sec.~\ref{sec:congress_network}} \label{app:con_weights}
%
Here, we describe our construction of an empirical network that encodes voting similarity among persons in the 117th U.S Senate.
We first downloaded voting-pattern data from VoteView \cite{lewis2018voteview}   to create a tensor with entries $\gamma_{ijk}=1$ when Senator $i$ and $j$ vote identically on bill $k$ and $\gamma_{ijk}=0$ otherwise. Letting $b_{ij}$ denote the number of bills in which both $i$ and $j$ vote, we define a weighted adjacency matrix with entries $\tilde{A}_{ij} = \frac{1}{b_{ij}}\sum_k \gamma_{ijk}   \in [0, 1]$ indicating for each  pair of Senators the fraction of bills on which they vote identically. Equivalently, we only sum over bills $k$ in which   both $i$ and $j$ cast a vote. Finally, we obtain a sparse network with an adjacency matrix ${\bf A}$ by applying a threshold  $\tau$ to   $\tilde{{\bf A}}$. That is, we define
\begin{equation}
    {A}_{ij} = \begin{cases} 
                        \tilde{A}_{ij}, & \tilde{A}_{ij} > \tau\\
                        0, & \text{otherwise}.
                      \end{cases}
\label{eq:threshold}
\end{equation}
After exploring a wide range of $\tau$ values, we selected $\tau = 0.40$, that is a minimum voting-together fraction of $0.40$. The resulting graph is a  weighted voting-similarity network \cite{waugh2009party,mucha2010communities,mucha2013polarization} that is undirected, and each edge $(i,j)$ has a weight $A_{ij}>0.4$.

\begin{figure}[t]
\centering
  \includegraphics[width=0.8\linewidth]{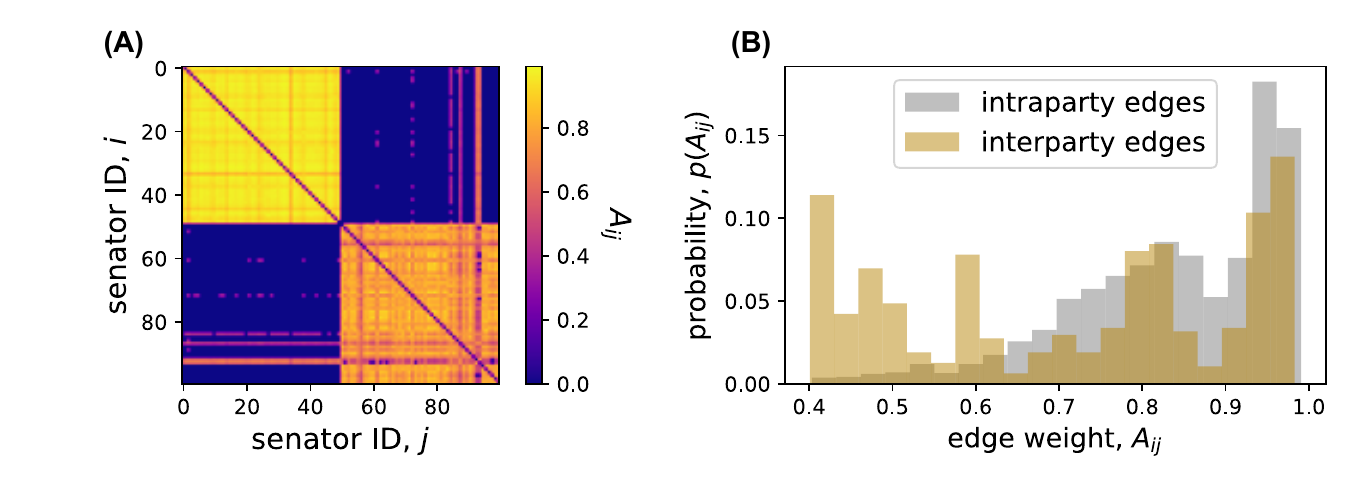}
  \vspace{-.3cm}
    \caption{{\bf Network encoding voting similarity among the 117th U.S. Senate.}
    \textbf{(A)}~Visualization of the adjacency matrix ${\bf A}$, where $A_{ij}$ encodes the fraction of bills in which Senators $i$ and $j$ vote identically (ignoring fractions less that $0.4$). 
    By sorting node ids by the Senator's political affiliations so that nodes $\{1,  \dots, 50\}$ are Republicans and nodes $\{51,  \dots, 100\}$ are either Democrats or independents, the resulting 2-community structure manifests as two large `blocks' in the matrix.
    \textbf{(B)}~Histogram depicts the probability distributions for edge weights $\{A_{ij}\}$, which we measure separately across intraparty and interparty edges.
    }
\label{fig:congress_weights}
\end{figure}

We note that there are other ways to construct voting-similarity networks, including, e.g., by restricting attention to non-unanimous roll calls \cite{waugh2009party} as opposed to considering all bills (as we have done).
Note that for convenience, we removed Kamala Harris and Kelly Loeffler from the data set because they  participated in the 117th Senate for less than a month before resigning or being replaced.  We also relabeled the independents to be Democrats since the current independents caucus with the Democrats and because, for simplicity, we would like to focus this experiment on a network having two communities of approximately equal size. 

Recall from Sec.~\ref{sec:congress_network}  that we would like to compare the rankings $\tilde{R}_{pq}$ of edges   within and between communities, i.e., intraparty and interparty edges. As such, we selected $\tau=0.4$ so that the resulting network has few edges between the two communities. In this case, we found 474 interparty edges and 2182 intraparty edges. This occurs due to party polarization; that is,  the vast majority of edges are intraparty edges that connect Senators having the same party affiliation, whereas relatively few  edges are interparty edges. Since most Senators are Republicans or Democrats (and the current independents largely act as Democrats in their roll call voting), the resulting network has two large communities. As shown in Fig.~\ref{fig:congress_weights}(A), these   communities manifest as 'blocks' in the adjacency matrix ${\bf A}$. In Fig.~\ref{fig:congress_weights}(B), we depict the distributions of edge weights $\{A_{ij}\}$, which we measure separately across intraparty and interparty edges.


\section{Further description of the multiplex brain network in Sec.~\ref{sec:brain_network}} \label{app:brain_weights}
Here, we provide additional discussion for our experiments with multiplex brain networks.  We first downloaded  adjacency matrices that encode network layers from \cite{brain_data}. See \cite{guillon2017loss} for a paper discussing this data and related experiments. Each adjacency matrix encodes spectral coherence  for a distinct particular frequency band.
Diagonal entries in the matrices are set to zero to ensure that there are no self-edges.  
For each network layer,  nodes represent brain regions and weighted edges represent coherence values.
The full data includes six network layers for each person, and there are data for 25 healthy persons and 25 persons with Alzheimers disease. We focused on two network layers that represent the frequency bands of $\delta$ and $\theta$ waves by examining two adjacency matrices  ${\bf A}^{(\delta)}$ and ${\bf A}^{(\theta)}$ for the first person in the data set. 

In Fig.~\ref{fig:brains}(A), we plot histograms of the coherence values ${ A}^{(\delta)}_{ij}$ and ${ A}^{(\theta)}_{ij}$ encoded in these two matrices. Because we would like to study sparse networks, we chose a threshold of $\tau=1.52$ and remove edges having a weight less than $\tau$. Observe that the two distributions $p({ A}^{(\delta)}_{ij})$ and $p({ A}^{(\theta)}_{ij})$ are similar, and so the two `sparsified' network layers have similar weight distributions and a similar number of edges.

In Fig.~\ref{fig:brains}(B), we present a plot that is similar to Fig.~\ref{fig:brain_fig}(B). Here, we plot the mean ranking $\langle \tilde{R}_{pq}\rangle_{inter}$ of interlayer edges versus timescale parameter $\beta$ for the coupling strength   $\omega = 1$. The different curves represent the 25 different healthy human brains in the study, and the dashed black curve gives the mean curve. Observe that all curves are qualitatively similar, suggesting that the effect of $\beta$ on edges' entropic importance is robust across the brain data.

\begin{figure}[h]
\centering
  \includegraphics[width=1\linewidth]{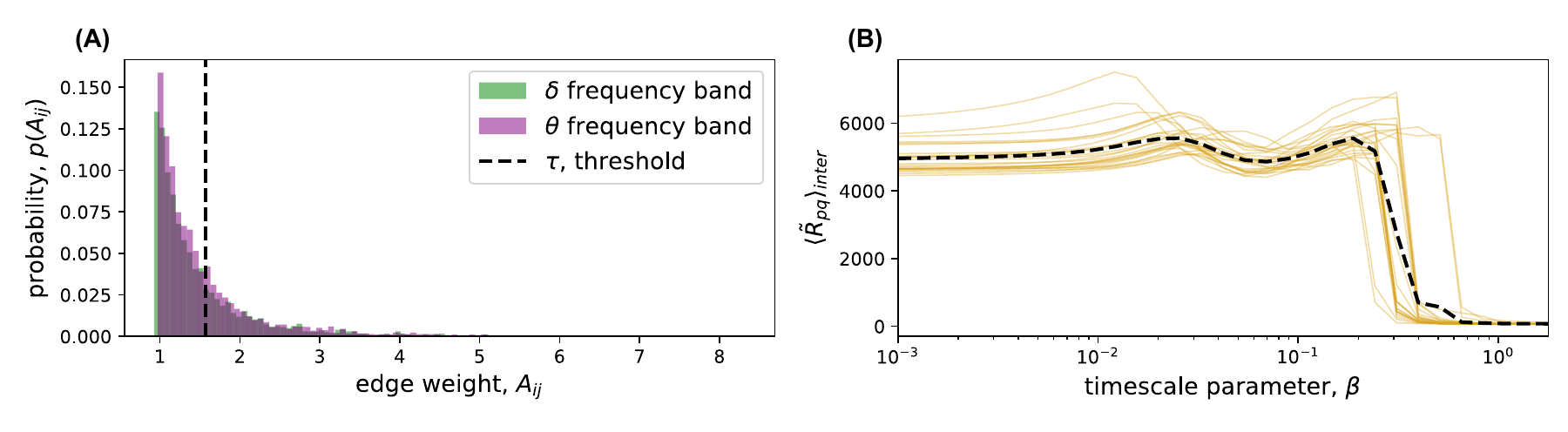}
  \vspace{-.8cm}
    \caption{{\bf Multiplex network encoding spectral coherence among brain regions.}
    \textbf{(A)}~Histograms depicting the probability distributions  edges weights $\{A_{ij}^{(\delta)}\}$ and $\{A_{ij}^{(\theta)}\}$ for two network layers encoding spectral-coherence relations among the fMRI activity brain regions at two frequency bands.
    \textbf{(B)}~We plot the mean ranking $\langle \tilde{R}_{pq}\rangle_{inter}$ of interlayer edges versus timescale parameter $\beta$ the coupling strength   $\omega = 1$. Different curves represent the 25 healthy human brains   in the data set \cite{brain_data,guillon2017loss}. The black dashed line depicts the mean curve.
    }
\label{fig:brains}
\end{figure}


\bibliographystyle{siam}
\bibliography{Entropy_bib}

\end{document}